\documentclass[a4paper]{article}

\usepackage{geometry}

\usepackage{palatino}
\usepackage{latexsym}

\usepackage[mathscr]{eucal}
\usepackage{amsmath}
\usepackage{amsthm}
\usepackage{amsfonts}
\usepackage{amssymb}
\usepackage{amscd}
\usepackage{color}
\usepackage{graphicx}
\usepackage{graphics}
\usepackage{pifont}
\usepackage{subfigure}
\usepackage[makeroom]{cancel}
\usepackage[normalem]{ulem}
\usepackage[dvipsnames]{xcolor}
\usepackage{tikz}

\theoremstyle{plain}
\newtheorem{theorem}{Theorem}
\newtheorem{remark}[theorem]{Remark}
\newtheorem{lemma}[theorem]{Lemma}
\newtheorem{proposition}[theorem]{Proposition}
\newtheorem{definition}[theorem]{Definition}

\newcommand{\ii}{\mathrm{i}}

\title{
Classical and quantum controllability of a rotating asymmetric molecule}

\begin{document}
\author{Eugenio Pozzoli\footnote{Inria, Sorbonne Universit\'e, Universit\'e de Paris, CNRS, Laboratoire Jacques-Louis Lions, Paris, France (eugenio.pozzoli@inria.fr).}}
\maketitle
\begin{abstract}
We study both the classical and quantum rotational dynamics of an asymmetric top molecule, controlled through three orthogonal electric fields that interact with its dipole moment. The main difficulties in studying the controllability of these infinite-dimensional quantum systems are the presence of severe spectral degeneracies in the drift Hamiltonian and the nonsolvability of the stationary free Schr\"odinger equation, which lead us to apply a perturbative Lie algebraic approach.

In this paper we show that, while the classical equations given by the Hamiltonian system on ${\rm SO}(3)\times \mathbb{R}^3$ are controllable for all values of the rotational constants and all dipole configurations, the Sch\"odinger equation for the quantum evolution on $L^2({\rm SO}(3))$ is approximately controllable for almost all values of the rotational constants if and only if the dipole is not parallel to any of the principal axes of inertia of the asymmetric rigid body.

\end{abstract}
\textbf{Keywords:} Schr\"odinger equation, quantum control, bilinear control systems, rotational dynamics, asymmetric top molecule, Euler equations
\section{Introduction}
The controllability problem of a quantum mechanical system has fundamental applications in chemistry, physics, computer science and engineering. From a mathematical point of view, this is often translated into the study of the controllability properties of the Schr\"odinger equation. Several different techniques have been developed in the last two decades in order to obtain results on this subject, and many models have been introduced as ideal playground for applications \cite{Altafini,Coron,
BGRS,BCCS,BCS,chambrion,Glaser2015,keyl,nersesyan,laurent,CS}. In this paper, we study the symmetries and the controllability of the Schr\"odinger partial differential equation on the Lie group of rotation ${\rm SO(3)}$. This system naturally describes the quantum rotational dynamics of a rigid body, that is interpreted as a symmetric or asymmetric rotating molecule. The wave function $\psi$ of this system is an element of the unit sphere of the Hilbert space $L^2({\rm SO}(3))$, and its evolution is governed by the Schr\"odinger equation
\begin{equation}\label{eq:top}
\ii\dfrac{\partial}{\partial t} \psi(R,t)= (AP_a^2+BP_b^2+CP_c^2)\psi(R,t)-\sum_{j=1}^3u_j(t)\langle  R \delta, e_j\rangle \psi(R,t),
\end{equation}
where $\psi(\cdot,t) \in L^2({\rm SO}(3))$, $AP_a^2+BP_b^2+CP_c^2$ is the rotational Hamiltonian (that is, the Laplace-Beltrami operator of $SO(3)$ w.r.t. the diagonal Riemannian metric ${\rm diag}(A,B,C)$ ), $A\geq B\geq C\geq 0$ are the rotational constants of the rigid body (related to the inertia moments through the identitites $2A=1/I_a,2B=1/I_b,2C=1/I_c$), $P_a,P_b,P_c$ are the angular momentum operators expressed w.r.t. the principal axes of inertia $a,b,c$ of the rigid body, and $-\langle  R \delta, e_i\rangle $ is the interaction Hamiltonian between the electric dipole moment $\delta$ of the molecule and the direction $e_i$ of the electric field, $i=1,2,3$ where $e_1,e_2,e_3$ is the canonical base of $\mathbb{R}^3$. The control law $u=(u_1,u_2,u_3)\in U$ is supposed to be smooth or piecewise constant and represents an electric field applied in the three orthogonal direction $e_1,e_2,e_3$, where $U\subset \mathbb{R}^3$ is a neighbourhood of the origin. Finally, $R \in {\rm SO}(3)$ is the matrix which describes the configuration of the rigid body in the space. 

\begin{figure}[ht!]
\subfigure[]{
\includegraphics[width=0.4\linewidth, draft = false]{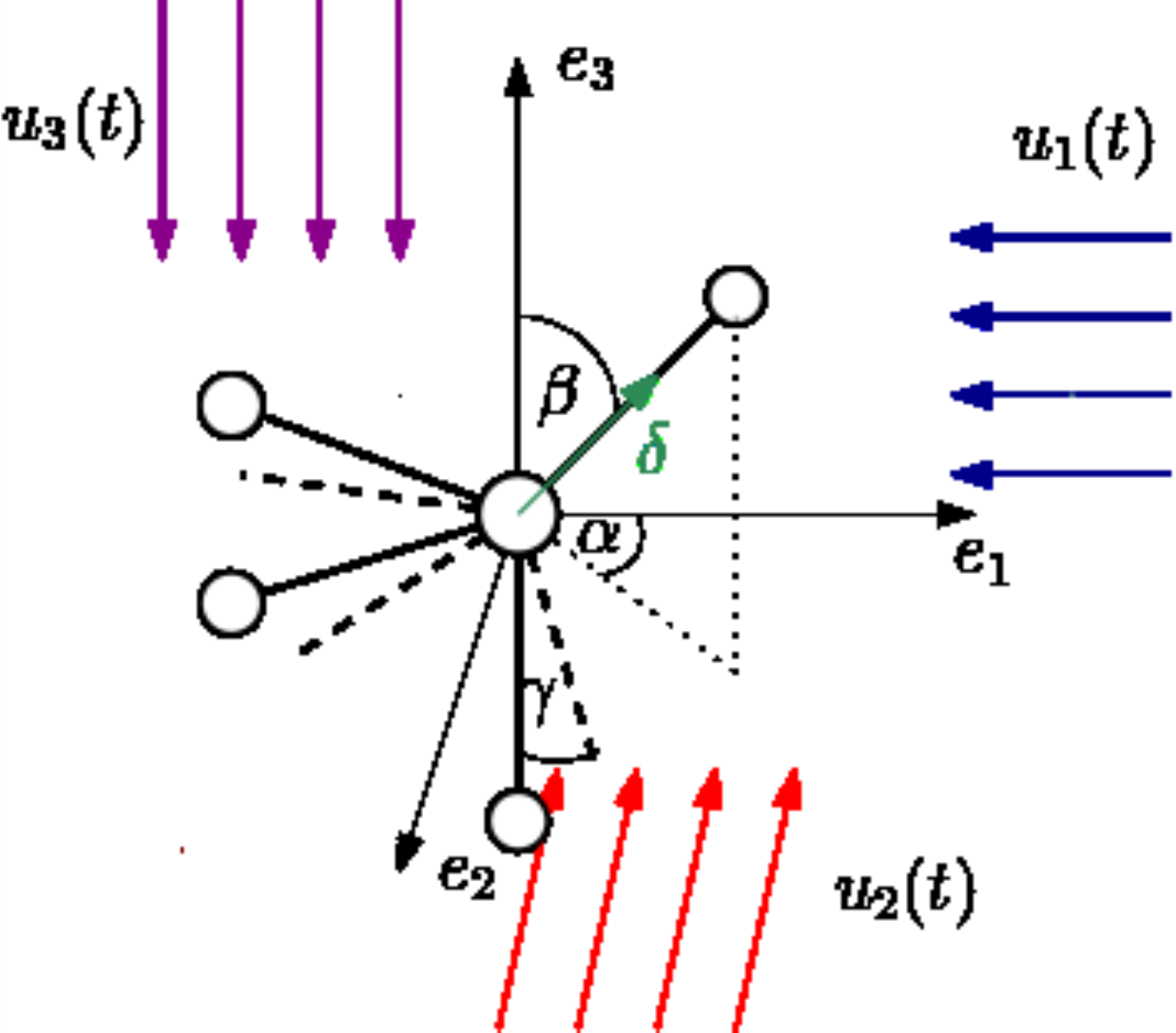} \label{symm} }\qquad\qquad\qquad
\subfigure[]{
\includegraphics[width=0.4\linewidth, draft = false]{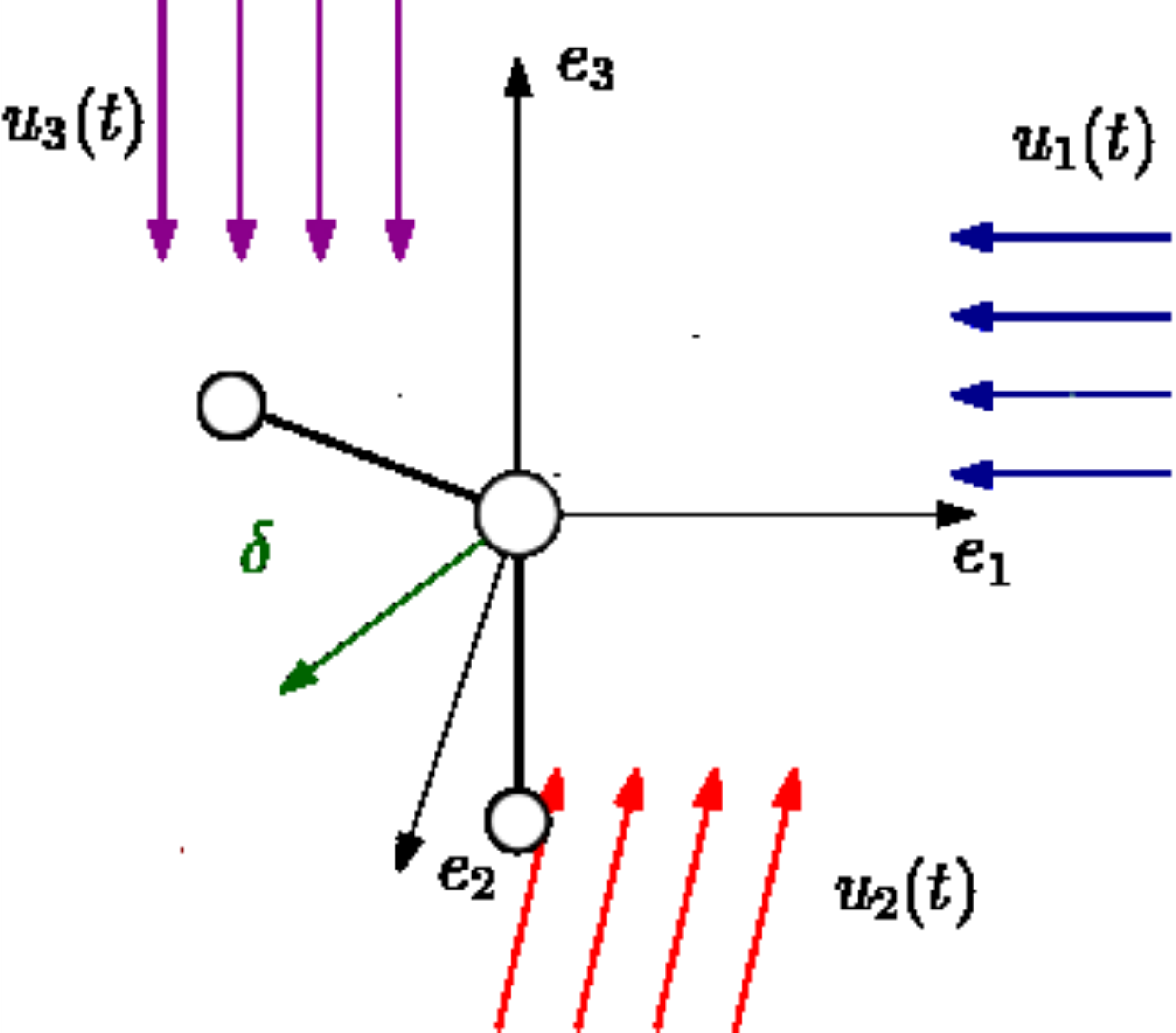} \label{asymm} }
\caption{Three orthogonal electric fields to control the rotation of \subref{symm} a symmetric molecule in $\mathbb{R}^3$ whose diagram represents, e.g., the chlorometane molecule $CH_3Cl$: its quantum rotation is not controllable, after \cite{Ugo-Mario-Io-symmetrictop}, as the electric dipole $\delta$ is parallel to the symmetry axis of the molecule; \subref{asymm} an asymmetric molecule in $\mathbb{R}^3$ whose diagram represents, e.g., the water molecule $H_2O$: its quantum rotation is not controllable, after Theorem \ref{theorem}(i), as the electric dipole $\delta=(0,0,\delta_c)^T$ is parallel to the axis of greatest inertia moment.}\label{fig:molecules}\end{figure}

 Molecules are extended objects and, under the rigid top approximation (which neglects the vibrations), are subject to the classifications in terms of their rotational constants $A\geq B\geq C\geq 0$: one distinghuishes asymmetric-tops ($A>  B>  C> 0$), prolate symmetric-tops ($A>  B= C> 0$), oblate symmetric-tops ($A= B>  C> 0$), spherical-tops ($A= B= C> 0$), and linear-tops ($A=B,\,C=0$).\\

 The general problem on whether molecular rotation is controllable goes back to the early days of quantum control: in the paper \cite{rabitz} crucial ideas were introduced and in particular a first proof of the approximate controllability of a rotating linear-top was presented. For a general overview on the controllability problem in molecular rotational dynamics we refer also to the review \cite{koch}, where the controllability problem for the Schr\"odinger evolution on ${\rm SO}(3)$ was proposed as an open problem, which is settled in this paper for almost all values of the inertia moments. 
 
 It is worth mentioning that, besides well-established applications in quantum chemistry such as microwave spectroscopy for determining molecular structure and controlling molecular reactivity, both from a theoretical \cite{rabitz-PRL-1992,Leibscher19} and an experimental \cite{dakin,PattersonNature13} point of view, rotational dynamics find new interesting applications in quantum information \cite{yu,victor}.
 
 One of the main feature of rotating molecules systems is that, even in the simpler case of a linear top, the spectrum exhibits severe increasing degeneracies at every eigenvalue. If we describe the linear top rotation with two quantum numbers $j\in\mathbb{N}$ and $m=-j,\dots,j$, which label the spherical harmonics $Y^j_m$ that are the eigenfunctions of the Laplace-Beltrami operator $\Delta_{S^2}$ on the two-sphere $S^2$ (the space of configurations of a linear molecule), each eigenvalue $E^j:=j(j+1)$ of $-\Delta_{S^2}$ is degenerate as it does not depend on $m$, with an associated eigenspace of dimension $2j+1$. For a symmetric top, an additional quantum number $k=-j,\dots,j$ is required, and corresponds to the discretization of the additional degree of freedom, which is the rotational motion about the symmetry axis. The harmonics of ${\rm SO}(3)$ are the Wigner $D$-functions $D^j_{k,m}$, for $j\in\mathbb{N}$ and $k,m=-j,\dots,j$: in particular, $D^j_{0,m}(\alpha,\beta,\gamma)=Y^j_m(\alpha,\beta)$, where $\alpha,\beta,\gamma$ denote the Euler angles as local coordinates of ${\rm SO}(3)$, and thus one recovers the linear top as a subsystem of the symmetric top by focussing on $k=0$, in analogy to the fact that $S^2$ can be recovered as the quotient space $SO(3)/S^1$, where $S^1$ denotes the group of rotation about the symmetry axis of the symmetric molecule. For symmetric molecules, the rotational eigenvalues have the following symmetry $E^j_{k}=E^j_{-k}$: besides the usual $(2j+1)$-dimensional degeneracies of the orientational quantum number $m$, also the quantum number $k$, for $k\neq 0$, has an additional $2$-dimensional degeneracy (see Figure \ref{fig:spectralsymmetrictop} for a picture of the spectral degeneracies of symmetric tops), which however vanishes in asymmetric tops. The physical explanation of the $m$-degeneracy is due to the orientational symmetry of rigid bodies: as in classical mechanics, also in quantum mechanics the rotational energy does not depend on the direction of the angular momentum. This complexity makes extremely hard the applications of techniques based on the existence of non-resonant spectral chains, developed to control infinite-dimensional discrete spectrum closed quantum systems in \cite{BCCS,BCMS,nersesyan}, which are applied to systems whose spectra are not too degenerate. Our spectral Lie algebraic technique can be applied to study the controllability problem on infinite-dimensional discrete spectrum closed quantum systems, such as \eqref{eq:top} or more generally the Schr\"odinger equation on a compact Riemannian manifold, and most importantly permits to treat drift Hamiltonians with severe degenerate spectra. The results are established by checking the controllability of an infinite family of overlapping finite-dimensional Galerkin approximations, together with non-resonant conditions on an infinite family of spectral gaps used to control the approximations. As a matter of fact, it allows to obtain approximate controllability results on linear, symmetric and asymmetric rotating tops.
 
  It is important to remark that, when the control operators are bounded, exact controllability never holds for the infinite-dimensional bilinear Schr\"odinger equation \cite{ball,turinici2,Chambrion-Caponigro-Boussaid-2020}, and one has to look for weaker properties such as approximate controllability.\\
 
  \begin{figure}[ht!]\begin{center}
\includegraphics[width=0.6\linewidth, draft = false]{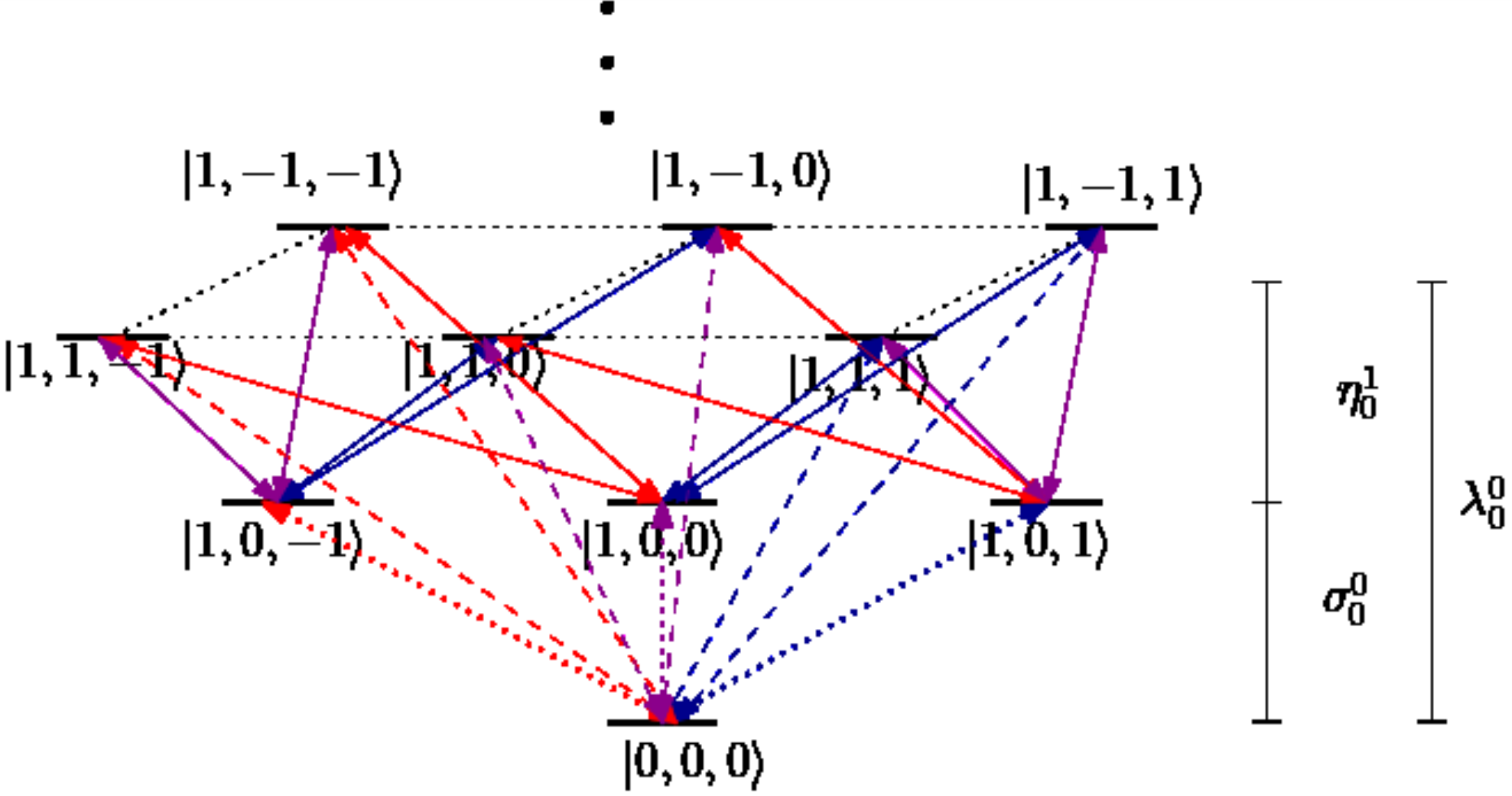}
\caption{Three-dimensional spectral graph associated with a symmetric-top, for $j=0,1$: transitions at frequencies $\lambda^0_0:=\mid E^1_1-E^0_0\mid $, $\sigma^0_0:=\mid E^1_0-E^0_0\mid $, and $\eta_0^1=\mid E^1_1-E^1_0\mid $ between the eigenstates $\mid j,k,m\rangle :=D^j_{k,m}$, driven by $H_1$ and $H_2$ (red and blue arrows), and $H_3$ (purple arrows). The dipole is not parallel nor orthogonal to the symmetry axis. Same-shaped arrows correspond to equal spectral gaps: the eigenstate $\mid 0,0,0\rangle $ corresponds to the eigenvalue $E^0_0$, the eigenstates $\mid 1,0,-1\rangle , \mid 1,0,0\rangle , \mid 1,0,1\rangle $ correspond to the eigenvalue $E^1_0$, the eigenstates $\mid 1,1,-1\rangle , \mid 1,1,0\rangle , \mid 1,1,1\rangle ,\mid 1,-1,-1\rangle , \mid 1,-1,0\rangle , \mid 1,-1,1\rangle $ correspond to the eigenvalue $E^1_{1}=E^1_{-1}$.} \label{fig:spectralsymmetrictop}\end{center}\end{figure}
 
 The approximate controllability (and stronger properties) of rotating linear tops (that is, $A=B,\, C=0$), modelled by the Schr\"odinger equation on the two-sphere $S^2$, has been established in \cite{BCS}, where the authors introduced a controllability test called the Lie-Galerking tracking condition. The extension to rotating symmetric tops has been obtained in the recent paper \cite{Ugo-Mario-Io-symmetrictop}, where a new version of the Lie-Galerkin tracking condition has been introduced and applied to classify the controllability of \eqref{eq:top} in the symmetric top cases (that is, when $A=B> C> 0$ or $A> B=C> 0$, under the nonresonant assumption $A/C\notin \mathbb{Q}$). In the present paper, we extend the results of \cite{Ugo-Mario-Io-symmetrictop} to almost every asymmetric molecule (that is, for a.e. $A> B> C> 0$). \\
 
One of the the main difficulties in proving controllability results for rotating asymmetric molecules is that the eigenvalue equation $H\psi=E\psi$, $E\in\mathbb{R}$, $\psi \in L^2({\rm SO}(3))$, has no explicit solution when $A> B> C> 0$, contrarily to the cases in which $A=B$ or $B=C$. This is one of the main differences w.r.t. rotating symmetric molecules studied in \cite{Ugo-Mario-Io-symmetrictop,symmtop_ifac}. In order to tackle this difficulty, we adopt a perturbative approach, using the fact that the rotational dynamics of an asymmetric top can be seen as analytic perturbations of those corresponding to two associated limiting oblate and prolate symmetric tops: when the electric dipole moment $\delta$ is not along any of the principal axes of inertia, this technique allows us to extend the approximate controllability from the symmetric cases to almost every value of the inertia moments, exploiting the stability of controllability results under the effect of an analytic perturbation. 
 
 The idea of studying the controllability of quantum systems in general configurations starting from symmetric cases (even if the latter have more degeneracies) has already been exploited, e.g., in \cite{panati,mehats}.
 
  On the other hand, when the dipole $\delta$ lies along any of the principal axes of the molecule, the structure of the control operators $\langle  R \delta, e_j\rangle , j=1,2,3$, in combination with some known symmetries of the asymmetric top eigenfunctions allow us to point out the existence of explicit invariant subspaces of \eqref{eq:top}.

The main result of this paper is a classification of the controllability of \eqref{eq:top} when $A> B> C> 0$:
 \begin{theorem}\label{theorem}
 System \eqref{eq:top} satisfies the following properties:
 \begin{itemize}
 \item[(i)] If $\delta\in\{(\delta_a,0,0)^T, (0,\delta_b,0)^T, (0,0,\delta_c)^T\}$, then \eqref{eq:top} is not controllable for all $A> B> C> 0$.
 \item[(ii)] If $\delta\notin\{(\delta_a,0,0)^T, (0,\delta_b,0)^T, (0,0,\delta_c)^T\}$, then \eqref{eq:top} is approximately controllable for almost every $A> B> C> 0$.
 \end{itemize}
 \end{theorem}
 In particular, while rotating linear tops are always approximately controllable, for symmetric and asymmetric tops invariant subspaces that prevent form controllability may arise, depending on the electric dipole moment orientation. The noncontrollable cases given in Theorem \ref{theorem}(i) are relevant, as there exist in nature very simple and fundamental asymmetric molecules that have electric dipole moment along one of the principal axes of inertia (e.g., the molecule of water, see Figure \ref{fig:molecules}(b)). Nonetheless, asymmetric molecules may be very complex objects and have in general dipole components along each of the three axes of inertia (e.g., the molecule of carvone and more broadly all chiral molecules): their rotations are hence (almost always) approximately controllable after Theorem \ref{theorem}(ii). \\
 
 We conclude this introduction with a remark on the different behaviours of classical and quantum systems. A conserved quantity of \eqref{eq:top} may not have a classical counterpart, that is, a corresponding conserved quantity for the associated Hamiltonian system on the cotangent bundle manifold ${\rm SO(3)}\times \mathfrak{so}(3)^*$. More precisely, in \cite{Ugo-Mario-Io-symmetrictop} it is proven that a symmetric molecule with dipole parallel to the symmetry axis (see Figure \ref{fig:molecules}(a)) has a classical and quantum conserved quantity (that is the component of the angular momentum along the symmetry axis), while a symmetric molecule with dipole orthogonal to the symmetry axis has a quantum conserved quantity but it is classically controllable. The discrepancy between classical and quantum controllability has already been observed in harmonic oscillator dynamics, which are classically but not quantum controllable \cite{rouchon}. Motivated by the discrepancy detected in \cite{Ugo-Mario-Io-symmetrictop} for symmetric tops, in this paper we also analyze the controllability of the classical equations of rotating asymmetric tops: in particular, we show that the quantum noncontrollable cases listed in Theorem \ref{theorem}(i) are in fact classically controllable.\\
 
 \medskip
 
 The paper is organized as follows: in Section \ref{sec:classicaltop} we interpret the Hamilton equations for the rotation of a rigid body as a control-affine system with recurrent drift, and in Theorem \ref{accidentallycla} we establish the classical controllability of every asymmetric rotating molecule controlled through three orthogonal electric fields, for every configuration of the electric dipole moment. In Section \ref{sec:quantumtop} we start by recalling an approximate controllability test for the discrete spectrum bilinear Schr\"odinger equation found in \cite{Ugo-Mario-Io-symmetrictop} (see Section \ref{sec:quantumcontrol}). We then classify the controllability of \eqref{eq:top} for a.e. $A> B> C> 0$. More in detail: in Theorem \ref{thm:symmetries-asymmtop} we show the existence of three invariant subspaces arising in \eqref{eq:top} when the dipole is parallel to any of the three principal axes of inertia and
in Theorem \ref{thm:asymmtopcontrol} 
we apply in a perturbative way the test previously introduced to show the approximate controllability of \eqref{eq:top} when the dipole is not parallel to any of the principal axes, for almost every value of the rotational constants.

\section{Classical controllability of asymmetric tops}\label{sec:classicaltop}
\subsection{Control-affine systems with recurrent drift}\label{sec:controlsystems}
Given the control-affine system 
\begin{equation}\label{control}
\dot{q}=X_0(q)+\sum_{i=1}^\ell u_i(t)X_i(q), \qquad q \in M,
\end{equation} 
on an $n$-dimensional smooth manifold $M$, with \emph{drift} $X_0$ and control fields $X_1,\dots,X_\ell$ (that are supposed to be $C^\infty$ vector fields on $M$), where the control functions $u=(u_1,\dots,u_\ell)$ are taken in $L^\infty(\mathbb{R},U)$ and $U\subset \mathbb{R}^\ell$ is 
a neighborhood of the origin, we denote
the \emph{reachable set} from $q_0\in M$ as the set
\begin{align*}
 \mathrm{Reach}(q_0):= \{& q \in M \mid  \exists \;  u,T \text{ s.t. the solution to (\ref{control}) with }\\ &q(0)=q_0 \text{ satisfies } q(T)=q \}. 
 \end{align*}
 \begin{definition}
System (\ref{control}) is said to be \emph{controllable} if $\mathrm{Reach}(q_0)=M$ for all $q_0\in M$.
\end{definition}
 When the drift $X_0$ is \emph{complete},
we say that it is \emph{recurrent} if for every open nonempty subset $V$ of $M$ and every time $t> 0$, there exists $\tilde{t}> t$ such that $\phi_{\tilde{t}}(V)\cap V \neq \emptyset$, where $\phi_{\tilde{t}}$ denotes the flow of $X_0$ at time ${\tilde{t}}$.
\subsection{The classical rotational dynamics of a molecule}\label{subsec:classicaldynamics}
Given a rigid body, the translational motion of its center of mass is decoupled from the rotational motion. We thus assume that the molecule can only rotate around its center of mass. To model the control problem for the rotation of a rigid body, one considers as manifold the tangent bundle $M={\rm SO}(3)\times \mathbb{R}^3$. We denote by $e_1,e_2,e_3$ a fixed orthonormal frame of $\mathbb{R}^3$ and by $a,b,c$ the principal axes of inertia of the asymmetric ridig body, with associated rotational constants $A> B> C> 0$ related to the inertia moments through the identitites $2A=1/I_a,2B=1/I_b,2C=1/I_c$. Both frames are attached to the rigid body's center of mass. The configuration of the molecule is identified with the unique matrix $R \in {\rm SO}(3)$ such that $R\;(v_a,v_b,v_c)^T=(v_1,v_2,v_3)^T$, where $(v_a,v_b,v_c)$ are the coordinates of a vector $v$ with respect to $a,b,c$, and $(v_1,v_2,v_3)$ are the coordinates of $v$ with respect to $e_1,e_2,e_3$, for any vector $v\in \mathbb{R}^3$. We assume that the electric charge of the molecule is modelled in dipole approximation with an electric dipole moment $\delta\in\mathbb{R}^3$ fixed inside the molecular frame.

Given the Hamiltonian function 
$$H=\left(AP_a^2+BP_b^2+CP_c^2\right)+V(R),\quad  V(R)=-\sum_{i=1}^3u_i\langle R\delta, e_i\rangle  $$
on ${\rm SO}(3)\times \mathbb{R}^3$ with coordinates $(R,P)=(R,(P_a,P_b,P_c))$,
the equations for the classical rotational dynamics of a molecule are the Hamilton equations associated to $H$, which read
\begin{equation}\label{euler1}
\begin{pmatrix}
\dot{R} \\ 
\dot{P}
\end{pmatrix}=X(R,P)+\sum_{i=1}^3u_i(t)Y_i(R,P), \quad (R,P)\in {\rm SO}(3) \times \mathbb{R}^3,\ u\in U,
\end{equation}
where
 \begin{equation}\label{fields}
X(R,P):=\begin{pmatrix}
R\,s(\rho P) \\
P \times (\rho P)
\end{pmatrix}, \quad Y_i(R,P):=\begin{pmatrix}
0\\
(R\delta)\times e_i
\end{pmatrix}, \quad i=1,2,3,
\end{equation}
$\rho P=(2AP_a,2BP_b,2CP_c)^T$ and $s$ denotes the isomorphism of Lie algebras \begin{equation}\label{eq:liealgebraisom}
s:(\mathbb{R}^3,\times) \rightarrow (\mathfrak{so}(3),[\cdot,\cdot]), \quad P=
\begin{pmatrix}
P_a \\
P_b\\
P_c
\end{pmatrix} \mapsto  s(P)=
\begin{pmatrix}
0 & -P_c & P_b\\
P_c & 0 & -P_a\\
-P_b & P_a & 0
\end{pmatrix}
\end{equation}
where $\times$ is the vector product. For a derivation of \eqref{euler1}, one can see e.g. 
\cite[Section 12.2]{jurdje} (where this is done for the heavy rigid body). We recall that $(u_1,u_2,u_3)\in U\subset \mathbb{R}^3$ and $U$ is such that $(0,0,0)\in {\rm Interior}(U)$.

System \eqref{euler1} can be seen as a control-affine system with drift $X$ and control fields $Y_1,Y_2,Y_3$.

If one uses quaternions $\mathbb{H}$ instead of the rotation group to parametrize the configuration of the rigid body, making use of the double covering map $S^3\rightarrow SO(3)$ (see \cite[Section 5.2]{Ratiu} for details), system \eqref{euler1} is lifted to  
$S^3\times \mathbb{R}^3$ with coordinates $(q,P)=((q_0,q_a,q_b,q_c),(0,P_a,P_b,P_c))\in\{q_0+\ii q_a+\mathrm{j}q_b+\mathrm{k}q_c\in \mathbb{H} \mid  q_0^2+q_a^2+q_b^2+q_c^2=1\}\times \{\ii P_a+\mathrm{j}P_b+\mathrm{k}P_c \in \mathbb{H}\mid  (P_a,P_b,P_c)\in \mathbb{R}^3\}=S^3\times \mathbb{R}^3\subset \mathbb{H}^2$. The lifted system read (here, $[P,\Omega]:=(P\Omega-\Omega P)$, where $P\Omega$ is the quaternion multiplication for any $P,\Omega \in \mathbb{R}^3\subset \mathbb{H}$)
\begin{equation}\label{quaternion}
\begin{cases}
\begin{aligned}
\dfrac{dq(t)}{dt}=&q(t)\rho P(t), \\ 
\dfrac{dP(t)}{dt}=&\frac{1}{2}[P(t), \rho P(t)]+\sum_{\mathrm{h}\in\{\ii,\mathrm{j},\mathrm{k}\}}\dfrac{u_\mathrm{h}(t)}{2}[\overline{q(t)}\mathrm{h}q(t),\delta].
\end{aligned}
\end{cases}
\end{equation}
\subsection{Controllability of classical asymmetric tops}
The classical rotational dynamics of asymmetric tops are always controllable, independently of their dipole moment.

\begin{theorem}\label{accidentallycla}
Let $A> B> C> 0$ and $\delta\neq (0,0,0)^T$. Then system \eqref{euler1} is controllable.
\end{theorem}
\begin{proof}
First of all, the drift $X$ is recurrent, as observed 
in \cite[Section 8.4]{AS}. So, we can use \cite[Theorem 5, Section 4.6]{jurdje} to prove that \eqref{euler1} is controllable. Thus, we need to show that 
$$\mathrm{dim}\Big( \mathrm{Lie}_{(R,P)}\{X,Y_1,Y_2,Y_3\}\Big)=6 \quad \forall (R,P) \in {\rm SO}(3) \times \mathbb{R}^3$$
where $\mathrm{Lie}_{(R,P)}\{X,Y_1,Y_2,Y_3\}$ denotes the Lie algebra generated by the vector fields $X,Y_1,Y_2,Y_3$ evaluated at $(R,P)$.
The structure of the proof is the following:
we are going to find six vector fields in
$\mathrm{Lie}\{X,Y_1,Y_2,Y_3\}$ whose span is six-dimensional everywhere but on a set of positive codimension, and we conclude by applying \cite[Lemma 2.2]{Ugo-Mario-Io-symmetrictop}.

 We remark that
$[X,Y_i](R,P)=\begin{pmatrix}
-R\,s(\rho[ \delta\times (R^{-1}e_i)]) \\
\star
\end{pmatrix}.$
 Denoting by $\Pi_{{\rm SO}(3)}$ the projection onto the ${\rm SO}(3)$ part of the tangent bundle, that is, $\Pi_{{\rm SO}(3)}: T({\rm SO}(3)\times \mathbb{R}^3)\rightarrow T{\rm SO}(3),$
we have
\begin{align*}
\mathrm{span}&\{\Pi_{{\rm SO}(3)}X(R,P),\Pi_{{\rm SO}(3)}[X,Y_1](R,P),  \Pi_{{\rm SO}(3)}[X,Y_2](R, P),\Pi_{{\rm SO}(3)}[X,Y_3](R,P) \} \\ &=R\,s\Big(\rho[\{\delta\}^\perp \oplus \mathrm{span}\{P\}]\Big).
\end{align*}
Hence, when $\langle  P, \delta \rangle  \neq 0$, one has
\begin{align}\nonumber
\dim & \Big( \mathrm{span}\{\Pi_{{\rm SO}(3)}X(R,P),\Pi_{{\rm SO}(3)}[X,Y_1](R,P),\Pi_{{\rm SO}(3)}[X,Y_2](R,P),\\ &\Pi_{{\rm SO}(3)}[X,Y_3](R,P) \} \Big) 
=3 .\label{3dim}
\end{align}

We now switch to the quaternion parametrization \eqref{quaternion}, more useful for computations. We have
\[
X(q,P)=\begin{pmatrix}
q\rho P \\
\frac{1}{2}[P,\rho P]
\end{pmatrix}=\begin{pmatrix}
-2q_aAP_a-2q_bBP_b-2q_cCP_c\\[1mm]
2q_0AP_a+2q_bCP_c-2q_cBP_b\\[1mm]
2q_0BP_b-2q_aCP_c+2q_cAP_a\\[1mm]
2q_0CP_c+2q_aBP_b-2q_bAP_a\\[1mm]
2(C- B)P_bP_c \\[1mm]
2(A- C)P_aP_c\\[1mm]
2(B- A)P_aP_b
\end{pmatrix}, \]
\[
Y_1(q,P)=\begin{pmatrix}
0\\
\frac{1}{2}[\overline{q}\ii q,\delta]
\end{pmatrix}=\begin{pmatrix}
0\\
0\\
0\\
0\\
(q_aq_b-q_0q_c)\delta_c-(q_aq_c+q_0q_b)\delta_b\\
(q_aq_c+q_0q_b)\delta_a-\frac{1}{2}(q_0^2+q_a^2-q_b^2-q_c^2)\delta_c\\
\frac{1}{2}(q_0^2+q_a^2-q_b^2-q_c^2)\delta_b-(q_aq_b-q_0q_c)\delta_a
\end{pmatrix},
\]
\[
 Y_2(q,P)=\begin{pmatrix}
0\\
\frac{1}{2}[\overline{q}\mathrm{j}q,\delta]
\end{pmatrix}=\begin{pmatrix}
0\\
0\\
0\\
0\\
\frac{1}{2}(q_0^2-q_a^2+q_b^2-q_c^2)\delta_c-(q_bq_c-q_0q_a)\delta_b\\
(q_bq_c-q_0q_a)\delta_a-(q_aq_b+q_0q_c)\delta_c\\
(q_aq_b+q_0q_c)\delta_b-\frac{1}{2}(q_0^2-q_a^2+q_b^2-q_c^2)\delta_a
\end{pmatrix}.
\]
We consider the six vector fields $X,Y_1,Y_2,[X,Y_1],[X,Y_2],[[X,Y_1],Y_1]$: the determinant of the 
matrix 
obtained by removing the first row from the $7\times 6$ matrix
$$(X(q,P),Y_1(q,P),Y_2(q,P),[X,Y_1](q,P),[X,Y_2](q,P),[[X,Y_1],Y_1](q,P))$$ 
is $D(q,P):=S_{\delta,(A,B,C)}(q)\langle  P,\delta\rangle $, where 
\begin{align*}
&S_{\delta,(A,B,C)}(q)=2ABCq_0S^{(0)}_{\delta}(q)\\ &\Big[(4ABC)^{-1}S^{(1)}_{\delta}(q)S^{(2)}_{\delta}(q)+(2A)^{-1}\Big((2B)^{-1}S^{(3)}_{\delta}(q)S^{(4)}_{\delta}(q)-(2C)^{-1}S^{(5)}_{\delta}(q)S^{(6)}_{\delta}(q)\Big)\Big] 
\end{align*}
and
\begin{align*}
S^{(0)}_{\delta}(q)&=\big[q_0 (-2 q_b \delta_a + 2 q_a \delta_b) + 
  2 q_c (q_a \delta_a + q_b \delta_b) + 
  q_0^2 \delta_c - (q_a^2 + q_b^2 - q_c^2) \delta_c\big]^2,\\
 S^{(1)}_{\delta}(q)&= q_0 q_b \delta_b + q_a q_c \delta_b - q_a q_b \delta_c + q_0 q_c \delta_c,\\
 S^{(2)}_{\delta}(q)&=2 q_0 \delta_a (q_c \delta_b - q_b \delta_c) - 
     2 q_a \delta_a (q_b \delta_b + q_c \delta_c)+ 
     q_0^2 (\delta_b^2 + \delta_c^2) + 
     q_a^2 (\delta_b^2 + \delta_c^2)  \\&-(q_b^2 + 
        q_c^2) (\delta_b^2 + \delta_c^2),   \\
   S^{(3)}_{\delta}(q)&=-2 q_a q_b \delta_a + 2 q_0 q_c \delta_a + q_0^2 \delta_b + 
        q_a^2 \delta_b - (q_b^2 + q_c^2) \delta_b,\\
        S^{(4)}_{\delta}(q)&=-2 (q_0 q_b + 
           q_a q_c) (\delta_a^2 + \delta_b^2) + (q_0^2 \delta_a + 
           q_a^2 \delta_a - (q_b^2 + q_c^2) \delta_a + 
           2 q_a q_b \delta_b - 2 q_0 q_c \delta_b) \delta_c,\\
           S^{(5)}_{\delta}(q)&=-2 (q_0 q_b + q_a q_c) \delta_a + (q_0^2 + q_a^2 - q_b^2 - 
           q_c^2) \delta_c,\\
           S^{(6)}_{\delta}(q)&=q_0^2 \delta_a \delta_b + 
        q_a^2 \delta_a \delta_b - (q_b^2 + 
           q_c^2) \delta_a \delta_b + 2 q_a q_c \delta_b \delta_c- 
        2 q_a q_b (\delta_a^2 + \delta_c^2)\\ & + 
        2 q_0 [q_b \delta_b \delta_c + 
           q_c (\delta_a^2 + \delta_c^2)].
  \end{align*}
Under the assumption $A> B> C> 0$, $S_{\delta,(A,B,C)}$ is not identically zero as long as $\delta\neq 0$. 
Hence, for all $(q,P) $ such that $D(q,P)\neq0$, 
\begin{align*}
\dim\Big(\mathrm{span}&\{X(q,P),Y_1(q,P),Y_2(q,P),[X,Y_1](q,P),[X,Y_2](q,P),\\ & [[X,Y_1],Y_1](q,P) \}\Big)=6,
\end{align*}
that is, outside the set $N:=
\{(q,P)\in S^3 \times \mathbb{R}^3\mid  D(q,P)=0\}$ the 
family $X,Y_1,Y_2$
is Lie bracket generating.

We are left to show that $\mathrm{Reach}(q,P)\not\subset N$ for every $(q,P) \in N$,
and then to apply \cite[Lemma 2.2]{Ugo-Mario-Io-symmetrictop}.
Let us start by considering the factor $\langle  P,\delta\rangle $ of $D$ and notice that, for any fixed $q \in S^3$, $Q:=\{P=(P_a,P_b,P_c)\in\mathbb{R}^3\mid \langle  P,\delta\rangle =0\}$ defines a surface inside $\{q\}\times \mathbb{R}^3$. Denote by $\Pi_{\mathbb{R}^3}:T(S^3\times \mathbb{R}^3)\to T\mathbb{R}^3$ the projection onto the $\mathbb{R}^3$ part of the tangent bundle.
The vector field $\Pi_{\mathbb{R}^3}X$ is tangent to $Q=\{\langle  P,\delta\rangle =0\}$ if and only if
 \begin{align*}
 \langle  \Pi_{\mathbb{R}^3} X\mid _Q, \nabla D\mid _Q\rangle = 0&\Leftrightarrow\langle  [P,\rho P]\mid _Q,\delta\rangle =0 \\\Leftrightarrow \delta \in \mathrm{span}\{P\mid _Q,\rho P\mid _Q\}&\Leftrightarrow \delta  \in \mathrm{span}\{\rho P\mid _Q\},
 \end{align*}
  where in the second equivalence we used that $\mathrm{rank}(P,\rho P)=2$ as long as the rotational constants do not satisfy $A=B=C$, and in the last equivalence we used that $\langle  P,\delta\rangle \mid _Q=0$. Then, we obtain that $\Pi_{\mathbb{R}^3} X$ is tangent to $Q$ if and only if 
$$P=t\begin{pmatrix}
\delta_a A^{-1} \\ \delta_b B^{-1} \\ \delta_c C^{-1}
\end{pmatrix}, \;\;t\in\mathbb{R}. $$
Using again that $\langle  P,\delta\rangle \mid _Q=0$, we see that $\langle  t(\delta_a A^{-1}  , \delta_b B^{-1}  ,\delta_c C^{-1} )^T,(\delta_a,\delta_b,\delta_c)^T\rangle =0,$ implying $t=0$. Finally, we have seen that  $\Pi_{\mathbb{R}^3} X$ is tangent to $Q$ if and only if $P=0$; as $\Pi_{\mathbb{R}^3}Y_i(q,P=0)\neq 0$, for any $i=1,2,3$, we conclude that the distribution spanned by $\{\Pi_{\mathbb{R}^3} X,\Pi_{\mathbb{R}^3} Y_i\}$ is not tangent to $Q$.
Putting things together, we have proved that
$$\mathrm{Reach}(q,P)\not\subset \{\langle  P,\delta\rangle =0\}, \quad  \forall (q,P) \in \{\langle  P,\delta\rangle =0\}.$$

We now conclude: if $(q,P) \in \{(q,P)\in S^3\times \mathbb{R}^3\mid  S_{\delta,(A,B,C)}(q)=0\}$ then we fix $P$ and we get two-dimensional strata $\{ q \in S^3\mid  S_{\delta,(A,B,C)}(q)=0\} \subset S^3$. The projections of the vector fields $X,[X,Y_1],[X,Y_2],[X,Y_3]$ on the base part of the bundle span a three-dimensional vector space if $\langle  P, \delta\rangle  \neq 0$, as observed in \eqref{3dim}. So, since we have previously shown that $P$ can be steered to a point such that $\langle  P, \delta\rangle  \neq 0$, trajectories can exit the set $\{ q \in S^3\mid  S_{\delta,(A,B,C)}(q)=0\}$. This concludes the proof of the theorem.
\end{proof}

\section{Quantum controllability of asymmetric tops}\label{sec:quantumtop}
\subsection{The discrete spectrum Schr\"odinger equation}
Let $\ell \in \mathbb{N}$ 
and $U\subset \mathbb{R}^\ell$ be a neighborhood of the origin. 
 Let $\mathcal{H}$ be an infinite-dimensional Hilbert space with scalar product $\langle  \cdot, \cdot \rangle  $ (linear in the first entry and conjugate linear in the second), $H$ be an unbounded self-adjoint operator with domain $\mathcal{D}(H)$, $H_1,\dots,H_\ell$ be bounded self-adjoint operators on $\mathcal{H}$. We consider the multi-input bilinear Schr{\"o}dinger equation
\begin{equation}\label{quantumcontrol}
\ii\frac{d\psi(t)}{dt}=(H+\sum_{j=1}^\ell u_j(t)H_j)\psi(t), \quad \psi(t) \in \mathcal{H}, \quad u(t) \in U.
\end{equation}
The main assumption is that $H$ has discrete spectrum with infinitely many distinct eigenvalues (possibly degenerate),
and we denote by $\mathcal{B}$ a Hilbert basis $\{\phi_k\}_{k\in \mathbb{N}}$ of $\mathcal{H}$ made of eigenvectors of $H$ associated with the family of eigenvalues $\{\lambda_k\}_{k\in \mathbb{N}}$.

Then, for every $(u_1,\dots,u_\ell)\in U$, $H+\sum_{j=1}^\ell u_jH_j$ generates a strongly continuous one-parameter group $e^{-\ii t(H+\sum_{j=1}^\ell u_jH_j)}$ of unitary operators on $\mathcal{H}$. One can therefore define the propagator $\Gamma_T^u$ at time $T$ of system (\ref{quantumcontrol}) associated with a 
piecewise constant control law $u(\cdot)=(u_1(\cdot),\dots,u_\ell(\cdot))$ by composition of flows of the type $e^{-\ii t(H+\sum_{j=1}^\ell u_jH_j)}$. 
\begin{definition}
\begin{itemize}
\item Given $\psi_0,\psi_1$ in the unit sphere $\mathcal{S}$ of $\mathcal{H}$, we say that $\psi_1$ is reachable from $\psi_0$ if there exist a time $T> 0$ and a piecewise constant control law $u:[0,T]\rightarrow U$ such that $\psi_1=\Gamma_T^u(\psi_0)$. We denote by $\mathrm{Reach}(\psi_0)$ the set of reachable points from $\psi_0$.
\item We say that (\ref{quantumcontrol}) is approximately controllable if for every $\psi_0\in \mathcal{S}$ the set $\mathrm{Reach}(\psi_0)$ is dense in $\mathcal{S}$.
\end{itemize}
\end{definition}
Equivalently, \eqref{quantumcontrol} is approximately controllable if for every  $\psi_0,\psi_1\in \mathcal{S}$ and every $\epsilon> 0$ there exists a piecewise constant control $u:[0,T]\rightarrow U$ such that $\| \Gamma_T^u(\psi_0)-\psi_1\| <  \epsilon. $
\subsection{An approximate controllability criterium}\label{sec:quantumcontrol}
Let $\{I_j\mid  j\in \mathbb{N}\}$ be a family of finite subsets of $\mathbb{N}$ such that $\cup_{j\in \mathbb{N}}I_j=\mathbb{N}$. 
Denote by $n_j$
the cardinality of $I_j$. Consider the subspaces 
\[\mathcal{M}_j:=\mathrm{span}\{\phi_n \mid  n\in I_j\}\subset \mathcal{H}\]
and their associated orthogonal projections 
\[\Pi_{\mathcal{M}_j}:\mathcal{H}\ni \psi \mapsto \sum_{n\in I_j}   \langle  \phi_n,\psi \rangle  \phi_n \in \mathcal{H}.
\]

 We project the control system \eqref{quantumcontrol} on $\mathcal{M}_j$ by defining the operators $H^{(j)}:= \Pi_{\mathcal{M}_j} H \Pi_{\mathcal{M}_j}$ and 
$H_i^{(j)}:= \Pi_{\mathcal{M}_j} H_i \Pi_{\mathcal{M}_j}$ 
for every $i=1,\dots,\ell$.
The set $\Sigma_j=\{\mid \lambda_l-\lambda_{l'}\mid \mid  l,l'\in I_j \}$ is then the collection of the spectral gaps of 
$\Pi_{\mathcal{M}_j} H \Pi_{\mathcal{M}_j}$.

For every $\sigma \geq 0$, and every $n$-dimensional matrix $M=(M_{l,k})_{l,k=1}^n$ (where $n$ is possibly $\infty$), let $\mathcal{E}_\sigma(M)$ be the $n$-dimensional matrix defined for every $l,k=1,\dots,n$ as
$$(\mathcal{E}_\sigma(M))_{l,k}=\begin{cases}
M_{l,k} & \text{ if } \mid \lambda_l-\lambda_k\mid =\sigma\\
0  & \text{ otherwise. }
\end{cases} $$
where $\delta_{l,k}$ is the Kronecker delta. The $n_j\times n_j$ matrix $\mathcal{E}_\sigma(H_i^{(j)})$ corresponds to the excitation in $H_i^{(j)}$ of the spectral gap $\sigma\in \Sigma_j$.

We introduce the sets
\begin{align*}
\Xi_j^0=\left\{(\sigma,i)\in \Sigma_j \times \{1,\dots,\ell\} \mid 
\mathcal{E}_\sigma(H_i)=\left[
\begin{array}{c| c| c}
0 & 0 &0\\
\hline
0 & \mathcal{E}_\sigma(H_i^{(j)}) &0\\
\hline
0 & 0 &0
\end{array}
\right]
 \right\},
\end{align*}
and
\[
\Xi_j^1=\left\{(\sigma,i)\in \Sigma_j \times \{1,\dots,\ell\} \mid  \mbox{}  \mathcal{E}_\sigma( H_i)=\left[
\begin{array}{c| c| c}
* & 0 &*\\
\hline
0 & \mathcal{E}_\sigma( H_i^{(j)}) &0\\
\hline
* & 0 &*
\end{array}
\right]  \right\},
\]
where the operator $\mathcal{E}_\sigma(H_i)$ is represented as an $\infty$-dimensional matrix w.r.t. the Hilbert basis $\{\phi_k\}_{k\in\mathbb{N}}$. 

Given $\sigma\in\Sigma_j, \sigma\neq 0$, such that $(\sigma,i)\in\Xi_j^1$, then transitions between eigenstates $\phi_l$ and $\phi_k$ that are resonant with $\sigma$ and coupled by $H_i$ (that is, $\sigma=\mid \lambda_l-\lambda_k\mid $ and $\langle  \phi_l,H_i\phi_k\rangle \neq 0$) are allowed as long as $\phi_l,\phi_k\in \mathcal{M}_j$ or $\phi_l,\phi_k\notin \mathcal{M}_j$ (so, they are not allowed if $\phi_l\in\mathcal{M}_j$ and $\phi_k\notin\mathcal{M}_j$ or $\phi_l\notin\mathcal{M}_j$ and $\phi_k\in\mathcal{M}_j$). If instead $(\sigma,i)\in\Xi_j^0$, then transitions between eigenstates $\phi_l$ and $\phi_k$ that are resonant with $\sigma$ and coupled by $H_i$  are allowed as long as $\phi_l,\phi_k\in \mathcal{M}_j$ (so, they are not allowed if $\phi_l,\phi_k\notin \mathcal{M}_j$, or $\phi_l\in\mathcal{M}_j$ and $\phi_k\notin\mathcal{M}_j$ or $\phi_l\notin\mathcal{M}_j$ and $\phi_k\in\mathcal{M}_j$).


We then define
\begin{equation}\label{eq:modes}
\nu_{j}^s:= \{\ii H^{(j)},\mathcal{E}_\sigma(\ii H_i^{(j)}) \mid  (\sigma,i)\in \Xi_{j}^s, \sigma \neq 0\}, \quad s=0,1.
\end{equation}
and notice that $\nu_{j}^0 \subset \nu_{j}^1 \subset \mathfrak{u}(n_j)$ (or $\mathfrak{su}(n_j)$ if $H^{(j)}$ and the $H_i^{(j)}$ are traceless).

 We denote by $\mathrm{Lie}(\nu_{j}^s)$ the Lie subalgebra  of $\mathfrak{u}(n_j)$ generated by the matrices in $\nu_{j}^s$, $s=0,1$,  and define 
$\mathcal{T}_j$ as the minimal ideal of $ \mathrm{Lie}(\nu_j^1)$ containing $\nu_j^0$.

Finally, we introduce the graph $\mathcal{G}$ with vertices $\mathcal{V}=\{I_j\mid  j\in \mathbb{N}\}$ and edges $\mathcal{E}=\{(I_j,I_k)\mid  j,k\in\mathbb{N},\;I_j\cap I_k\neq \emptyset\}$. We have the following test for the approximate controllability of \eqref{quantumcontrol}:
\begin{theorem}(\cite{Ugo-Mario-Io-symmetrictop})\label{LGTC}
If the graph $\mathcal{G}$ is connected and $\mathfrak{su}(n_j)\subset\mathcal{T}_j$ for every $j\in \mathbb{N}$, then \eqref{quantumcontrol} is approximately controllable.
\end{theorem}
\begin{remark}
As a byproduct of the geometric control theory behind the proof of Theorem \ref{LGTC}, one actually obtains slightly stronger controllability results with the same assumptions: e.g., exact controllability in projections \cite{CS}, tracking results \cite{BCS}, approximate controllability of the density matrix \cite{BCCS,BCS} and in finer $H^s$-topologies \cite{weaklycoupled,BCS}. It is also worth mentioning that Theorem \ref{LGTC} can be applied even in the presence of unbounded control operators $H_i$, and the approximate controllability can be obtained also with different classes of control laws $u$ (e.g., smooth controls) \cite{CS}. For a well-defined notion of solution with unbounded control operators and weak classes of control laws, see also \cite{Chambrion-Caponigro-Boussaid-2020}.  
\end{remark}

\subsection{The Schr{\"o}dinger equation of a rotating molecule}\label{sec:quantumrotatingtop}

We use Euler's angles $(\alpha,\beta,\gamma)\in[0,2\pi)\times[0,\pi]\times [0,2\pi)$ to parametrize the configuration space ${\rm SO}(3)$ of the molecule. As we have already recalled in the description of a classical rotating rigid body, the space fixed frame $e_1,e_2,e_3$ is related to the body fixed frame $a,b,c$ (made of principal axes of inertia, chosen such that the rotational constants satisfy $A> B> C$) through a rotation $R\in SO(3)$. We adopt the following convention for the Euler's angles: $R(\alpha,\beta,\gamma)=R_{e_3}(\alpha)R_{e_2}(\beta)R_{e_3}(\gamma)$, that is, $R$ is a composition of three rotations, where $R_{e_i}(\theta)\in {\rm SO}(3)$ is the rotation of angle $\theta$ about the axis $e_i$. 

The explicit expression of the matrix $R(\alpha,\beta,\gamma)\in {\rm SO}(3)$ is 
\begin{equation}\label{rotation}
R=\begin{pmatrix}
\cos\alpha \cos\beta \cos\gamma-\sin\alpha \sin\gamma & -\cos\alpha \cos\beta \sin\gamma-\sin\alpha \cos\gamma & \cos\alpha \sin\beta\\
\sin\alpha \cos\beta \cos\gamma+\cos\alpha \sin\gamma &-\sin\alpha \cos\beta \sin\gamma+\cos\alpha \cos\gamma & \sin\alpha \sin\beta\\
-\sin\beta \cos\gamma& \sin\beta \sin\gamma& \cos\beta
\end{pmatrix}.
\end{equation}

The angular momentum operators w.r.t. the body fixed frame are given by the vector fields on ${\rm SO}(3)$
\begin{equation}\label{eq:P_i}
\begin{cases}
\begin{aligned}
P_1&=-\ii \cos\gamma\frac{\cos\beta}{\sin\beta}\frac{\partial }{\partial \gamma}+\ii \frac{\cos\gamma}{\sin\beta}\frac{\partial }{\partial \alpha}-\ii \sin\gamma \frac{\partial }{\partial \beta}  ,\\
P_2&=\ii\sin\gamma\frac{\cos\beta}{\sin\beta}\frac{\partial }{\partial \gamma}-\ii \frac{\sin\gamma}{\sin\beta}\frac{\partial }{\partial \alpha} -\ii \cos\gamma \frac{\partial }{\partial \beta} , \\
P_3&= -\ii  \frac{\partial }{\partial \gamma}.
\end{aligned}
\end{cases}
\end{equation}
which are seen as unbounded differential self-adjoint operators acting on the Hilbert space $L^2({\rm SO}(3))$ (we consider $L^2({\rm SO(3)})$ endowed with the $L^2$-scalar product w.r.t. the Haar measure $\frac{1}{8}d\alpha d\gamma \sin\beta d\beta$ of ${\rm SO}(3)$). By defining the vector $\vec{P}=(P_1,P_2,P_3)$, one can represent the angular momentum w.r.t. the space fixed frame through a rotation: $\vec{J}:=R\vec{P}$. In particular, $J_3=-\ii\partial/\partial \alpha$.


The operator $P_3$ is the angular momentum component along the axis of quantization that conventionally is identified with the symmetry axis for a symmetric rigid body. We identify each of the three components $P_a,P_b,P_c$ with one of the $P_i$, $i=1,2,3$. We are going to need two different identifications, (i) if the molecule is oblate symmetric, that is $A=B> C> 0$, the symmetry axis of the rigid body is $c$ and we set $P_3=P_c$ and $P_1=P_b, P_2=P_a$: of course, this choice can be made also if the molecule is asymmetric, and is called the oblate convention \cite[Table 7.3]{gordy}; (ii) if the molecule is prolate symmetric, that is $A> B=C> 0$, the symmetry axis of the rigid body is $a$ and we set $P_3=P_a$, and $P_1=P_b,P_2=P_c$: this choice is called the prolate convention \cite[Table 7.3]{gordy}.


 The rotational Hamiltonian of a molecule is \cite[Chapter 7]{gordy} $$H=AP_a^2+BP_b^2+CP_c^2,$$ which is an unbounded self-adjoint differential operator acting on $L^2({\rm SO}(3))$. The interaction Hamiltonian between the electric dipole moment $\delta=(\delta_a,\delta_b,\delta_c)^T$ and the electric field in the direction $e_i$, is given by the Stark effect \cite[Chapter 10]{gordy}
$$H_i(\alpha,\beta,\gamma)=-\langle  R(\alpha,\beta,\gamma) \delta, e_i\rangle,\quad i=1,2,3,$$ seen as a bounded self-adjoint operator of multiplication acting on $L^2({\rm SO}(3))$. 
The rotational Schr{\"o}dinger equation for a rigid molecule subject to three orthogonal 
electric fields reads
\begin{equation}\label{schro}
\ii\dfrac{\partial}{\partial t} \psi(\alpha,\beta,\gamma;t)= H\psi(\alpha,\beta,\gamma;t)+\sum_{l=1}^3u_l(t)H_l(\alpha,\beta,\gamma)\psi(\alpha,\beta,\gamma;t), 
\end{equation}
with $\psi(t) \in L^2({\rm SO}(3))$ and $u(t) \in U$, for some neighborhood $U$ of $0$ in $\mathbb{R}^3$.

\subsection{Harmonics on asymmetric tops} In this section we recall some facts on angular momentum theory, following \cite{Varshalovich,gordy}.
Since $H$ is the Laplace-Beltrami operator of a compact Riemannian manifold, there exists an orthonormal Hilbert basis of $L^2({\rm SO}(3))$ made of eigenfunctions of $H$.

 When $A=B$ or $B=C$ this basis is made by the so-called Wang functions \cite[Section 7.2]{gordy}
\begin{equation}\label{eq:wangfunctions}
S_{0,m,0}^j:=D_{0,m}^j,\qquad 
S_{k,m,p}^j:=\dfrac{1}{\sqrt{2}}\left(D_{k,m}^j+(-1)^p D_{-k,m}^j\right), \quad k=1,\dots,j,
\end{equation}
for $j\in \mathbb{N}$, $m=-j,\dots,j$, and $p=0,1$, where 
\begin{equation}\label{explicit}
 D_{k,m}^j(\alpha,\beta,\gamma):=e^{\ii (k\gamma+m\alpha)}d_{k,m}^j(\beta), \qquad j\in \mathbb{N},\quad k,m=-j,\dots,j,
 \end{equation}
are the Wigner $D$-functions and the $d^j_{k,m}$ satisfy a suitable Legendre differential equation \cite{Varshalovich}. 
If the molecule is symmetric, e.g. oblate, by imposing the symmetry relation $A=B$ one can write the rotational Hamiltonian in the following form
$$H(A,A,C)=AJ^2-(A-C)P_c^2,$$
and the eigenvalues of $H$ are given by \cite{gordy}
\begin{equation}\label{spectrum}
HS_{k,m,p}^j=(Aj(j+1)-(A-C)k^2)S_{k,m,p}^j=:E_k^j(A,C)S_{k,m,p}^j.
\end{equation}
for all $j\in\mathbb{N}$, $m=-j,\dots,j$, $k=0,\dots,j$, $p=0,1$. For prolate symmetry (i.e., $B=C$) one analogously obtains $H(A,C,C)=CJ^2+(A-C)P_a^2$,
with eigenvalues $HS_{k,m,p}^j=(Cj(j+1)+(A-C)k^2)S_{k,m,p}^j=:E_k^j(C,A)S_{k,m,p}^j$.
Since the eigenvalues 
of $H$ do not depend on $m$, the energy level $E_k^j$ is $(2j+1)$-degenerate with respect to $m$.
Moreover, since they do not depend on $p$, when $k \neq 0$ the energy level $E_k^j$ is also $2$-degenerate w.r.t. $p$. This extra $p$-degeneracy is common to all symmetric molecules, and vanishes in asymmetric molecules.

When $A> B> C> 0$, the $S^j_{k,m,p}$ are no longer eigenfunctions of $H(A,B,C)$, but one can still express its eigenfunctions as linear combinations of the $S^j_{k,m,p}$. More precisely, known commutation properties of $H$ allow to write any eigenfunction as \cite[Chapter 7.2]{gordy}
\begin{equation}\label{eigenfunctions_bis}
\Psi^j_{\tau,m}(A,B,C)=\sum_{\substack{ k=0,\dots, j\\ k \text{ only even or only odd}\\ p=0 \text{ only or } p=1 \text{ only}}}z^{j,\tau}_{k,m,p}(A,B,C) S^j_{k,m,p},
\end{equation}
for some coefficients $z^{j,\tau}_{k,m,p}(A,B,C)\in\mathbb{C}$. That is, the numbers $j$,$m$,$p$ and the parity of $k$ are still well-defined.  
Moreover, it is also known that the eigenvalues of $H(A,B,C)$ are still $(2j+1)$-degenerate w.r.t $m$, while the $p$-degeneracy is lifted (see, e.g., \cite[Chapter 7]{gordy}): the spectrum of $H(A,B,C)$ is the set of eigenvalues $\{E^j_\tau(A,B,C) \mid  j\in\mathbb{N}, \tau=-j,\dots,j\}$, where
$$H(A,B,C)\Psi^j_{\tau,m}(A,B,C)=E^j_\tau(A,B,C) \Psi^j_{\tau,m}(A,B,C)\,,\quad m=-j,\dots,j,$$
and each $E^j_\tau(A,B,C)$ is degenerate, corresponding (generically) to the $(2j+1)$-dimensional eigenspace spanned by $\{\Psi^j_{\tau,m}(A,B,C)\mid  m=-j,\dots,j\}$.
\begin{remark}
 The knowledge of the coefficients $z^{j,\tau}_{k,m,p}$ and the eigenvalues $E^j_\tau$ is useful for controllability: in \cite{monika_eugenio}, they are obtained by diagonalizing the truncations of $H$, and controllability results on finite-dimensional approximations of asymmetric chiral molecules are obtained in this way; in the present paper, in order to treat the full rotational spectrum, we shall obtain informations on them in a perturbative way.
 \end{remark}



\subsection{Noncontrollable dipole configurations of quantum asymmetric tops}
Thanks to the structure \eqref{eigenfunctions_bis} of the eigenfunctions of the asymmetric-top, we can point out the existence of invariant subspaces for \eqref{schro}, when the dipole is parallel to any of the principal axes of inertia. We introduce the Hilbert subspaces
 \begin{align*}
\mathcal{K}_{e(o)}&:=\overline{\mathrm{span}}\{S_{k,m,p}^j \mid  j\in \mathbb{N},\;m=-j,\dots,j,\;k=0,\dots,j,\;p=0,1,\\ & k \text{ even(odd)}\}, \\ 
 \mathcal{G}_{e(o)}&:=\overline{\mathrm{span}}\{S_{k,m,p}^j \mid  j\in \mathbb{N},\;m=-j,\dots,j,\;k=0,\dots,j,\;p=0,1,\\ & j+p \text{ even(odd)}\}, \\
 \mathcal{L}_{e(o)}&:=\overline{\mathrm{span}}\{S_{k,m,p}^j \mid  j\in \mathbb{N},\;m=-j,\dots,j,\;k=0,\dots,j,\;p=0,1,\\ &j+k+p \text{ even(odd)}\},
 \end{align*}
which in particular give three different orthogonal decompositions of the ambient space $$L^2({\rm SO}(3))=\mathcal{K}_{e}\oplus \mathcal{K}_{o}=\mathcal{G}_{e}\oplus \mathcal{G}_{o}=\mathcal{L}_{e}\oplus \mathcal{L}_{o}. $$
\begin{theorem}\label{thm:symmetries-asymmtop}
Let $A\geq B \geq C> 0$. \begin{enumerate}
\item[(i)]If $\delta=(0,0,\delta_c)^T$, then the spaces $\mathcal{K}_{e}$ and $\mathcal{K}_{o}$ are invariant for the propagators of \eqref{schro}. In particular, \eqref{schro} is not controllable.
\item[(ii)]If $\delta=(0,\delta_b,0)^T$, then the spaces $\mathcal{L}_{e}$ and $\mathcal{L}_{o}$ are invariant for the propagators of \eqref{schro}. In particular, \eqref{schro} is not controllable.
\item[(iii)]If $\delta=(\delta_a,0,0)^T$, then the spaces $\mathcal{G}_{e}$ and $\mathcal{G}_{o}$ are invariant for the propagators of \eqref{schro}. In particular, \eqref{schro} is not controllable.
\end{enumerate}
\end{theorem}
\begin{proof}
We adopt the oblate convention, so that $(\delta_1,\delta_2,\delta_3)^T=(\delta_b,\delta_a,\delta_c)^T$. To prove (i), we claim that
\begin{equation}\label{selection1}
 \delta=(0,0,\delta_3)^T\;\;\Rightarrow\;\;\langle  S^j_{k,m,p},H_i S^{j'}_{k',m',p'}\rangle =0\quad \text{ if }k\neq k',\;\;i=1,2,3.
 \end{equation}
 This is easily seen by noticing that the expressions of the control operators $H_1,H_2,H_3$ do not depend on the angle $\gamma$ if $\delta_1=\delta_2=0$: indeed, using \eqref{rotation}, $H_1=-(\cos\alpha\sin\beta)\delta_3$, $H_2=-(\sin\alpha\sin\beta)\delta_3$, and $H_3=-\cos\beta\delta_3$. Thus, 
 \begin{align*}
 &\langle  D_{k,m}^j, H_i  D_{k',m'}^{j'}\rangle _{L^2({\rm SO}(3))}\\ 
 =&\langle e^{\ii k \gamma},e^{\ii k'\gamma}\rangle_{L^2(S^1)}\langle e^{\ii m\alpha}d^j_{k,m}(\beta),H_i(\alpha,\beta)e^{\ii m'\alpha}d^{j'}_{k',m'}(\beta)\rangle_{L^2(S^2)}=0 
 \end{align*}

if $k\neq k'$, for $i=1,2,3$, using the orthogonality of the functions $e^{ik\gamma}$ and $e^{ik'\gamma}$ for $k\neq k'$. Hence, \eqref{eigenfunctions_bis} plus \eqref{selection1} imply
$$ \delta=(0,0,\delta_3)^T\;\;\Rightarrow  H_i \mathcal{K}_{e(o)}\subset \mathcal{K}_{e(o)},\;\;i=1,2,3,$$
and moreover \eqref{eigenfunctions_bis} implies that $H\mathcal{K}_{e(o)}\subset \mathcal{K}_{e(o)}$.

In order to prove (ii), we claim that 
\begin{equation}\label{selection2}
\delta=(\delta_1,0,0)^T\Rightarrow\langle  S^j_{k,m,p},H_i S^{j'}_{k',m',p'}\rangle =0 \text{ if }j+k+p\not\equiv j'+k'+p'\mathrm{mod}\;2, \;i=1,2,3.
 \end{equation}
Then, \eqref{eigenfunctions_bis} plus \eqref{selection2} imply
$$ \delta=(\delta_1,0,0)^T\;\;\Rightarrow  H_i \mathcal{L}_{e(o)}\subset \mathcal{L}_{e(o)},\;\;i=1,2,3,$$
and moreover \eqref{eigenfunctions_bis} implies $H\mathcal{L}_{e(o)}\subset \mathcal{L}_{e(o)}$. So, we are left to prove the claim \eqref{selection2}: first, we recall the selection rules \cite[Table 2.1]{gordy}
\begin{equation}\label{rules}
\langle  D_{k,m}^j , \ii H_l D_{k',m'}^{j'} \rangle =0,
\end{equation}
when $\mid j'-j\mid > 1$, or $\mid k'-k\mid >1$ or $\mid m'-m\mid > 1$, for every $l=1,2,3$.

 Moreover, the non-vanishing matrix elements for the interaction Hamiltonians in the $D^j_{k,m}$-basis are given by (see, e.g., \cite[Table 2.1]{gordy})
\begin{equation}\label{kk+1}
\begin{cases}
\langle  D_{k,m}^j, \ii H_1 D_{k+1,m\pm1}^{j+1} \rangle &=-c_{j,k,\pm m}(\delta_2+\ii\delta_1),   \\
\langle  D_{k,m}^j , \ii H_1 D_{k-1,m\pm1}^{j+1} \rangle &=c_{j,-k,\pm m}(\delta_2-\ii\delta_1),  \\
\langle  D_{k,m}^j , \ii H_2 D_{k+1,m\pm1}^{j+1} \rangle &= \mp \ii c_{j,k,\pm m} (\delta_2+\ii\delta_1),   \\
\langle  D_{k,m}^j , \ii H_2 D_{k-1,m\pm1}^{j+1} \rangle &=\pm \ii c_{j,-k,\pm m}(\delta_2-\ii\delta_1),  \\
\langle  D_{k,m}^j , \ii H_3 D_{k\pm1,m}^{j+1} \rangle &=\pm \ii d_{j,\pm k,m}(\delta_2\pm\ii\delta_1),  
\end{cases}
\end{equation}
where 
\[
\begin{split}
 &c_{j,k,m}:=  \dfrac{ [(j+k+1)(j+k+2)]^{1/2}[(j+m+1)(j+m+2)]^{1/2}}{4(j+1)[(2j+1)(2j+3)]^{1/2}},\\ &
 d_{j,k,m}:=  \dfrac{ [(j+k+1)(j+k+2)]^{1/2} [(j+1)^2-m^2]^{1/2}}{2(j+1)[(2j+1)(2j+3)]^{1/2}},
 \end{split}
 \]
 then
\begin{equation}\label{jj}
\begin{cases}
\langle  D_{k,m}^j , \ii H_1D_{k+1,m\pm1}^{j} \rangle &=\mp h_{j,k,\pm m}(\delta_2+\ii \delta_1),  \\
\langle  D_{k,m}^j , \ii H_1D_{k-1,m\pm 1}^{j} \rangle &=\mp h_{j,-k,\pm m}(\delta_2-\ii \delta_1),  \\
\langle  D_{k,m}^j , \ii H_2D_{k+1,m\pm1}^{j} \rangle &= -\ii h_{j,k,\pm m}(\delta_2+\ii \delta_1),  \\
\langle  D_{k,m}^j , \ii H_2D_{k-1,m\pm1}^{j} \rangle &=-\ii h_{j,-k,\pm m}(\delta_2- \ii \delta_1),   \\
\langle  D_{k,m}^j , \ii H_3D_{k\pm1,m}^{j} \rangle &=-\ii q_{j,\pm k,m}(\delta_2 \pm \ii \delta_1), 
\end{cases}
\end{equation}
where 
\begin{align*}
 h_{j,k,m}&:=  \dfrac{[j(j+1)-k(k+1)]^{1/2}[j(j+1)-m(m+1)]^{1/2}}{4j(j+1)}, \\
q_{j,k,m}&:=  \dfrac{[j(j+1)-k(k+1)]^{1/2}m}{2j(j+1)},
\end{align*}
then
\begin{equation}\label{kk}
\begin{cases}
\langle  D_{k,m}^j , \ii H_1D_{k,m\pm 1}^{j+1} \rangle &=a_{j,k,\pm m}\delta_3,  \\
\langle  D_{k,m}^j , \ii H_2D_{k,m\pm 1}^{j+1}\rangle &=\pm \ii a_{j,k,\pm m}\delta_3,  \\
\langle  D_{k,m}^j , \ii H_3D_{k,m}^{j+1} \rangle &=-\ii b_{j,k,m}\delta_3,  
\end{cases}
\end{equation}
where
\begin{align*}
a_{j,k,m}&:= \dfrac{[(j+1)^2-k^2]^{1/2}[(j+m+1)(j+m+2)]^{1/2}}{2(j+1)[(2j+1)(2j+3)]^{1/2}},\\
b_{j,k,m}&:=  \dfrac{[(j+1)^2-k^2]^{1/2}[(j+1)^2-m^2]^{1/2}}{(j+1)[(2j+1)(2j+3)]^{1/2}},
\end{align*}
and finally
\begin{equation}\label{jjkk}
\begin{cases}
\langle  D_{k,m}^j , \ii H_1D_{k,m\pm 1}^{j} \rangle &=\pm f_{j,k,\pm m}\delta_3,  \\
\langle  D_{k,m}^j , \ii H_2D_{k,m\pm 1}^{j}\rangle &= \ii f_{j,k,\pm m}\delta_3,  \\
\langle  D_{k,m}^j , \ii H_3D_{k,m}^{j} \rangle &=\ii g_{j,k,m}\delta_3, 
\end{cases}
\end{equation}
where
\begin{align*}
f_{j,k,\pm m}&:= \dfrac{k[j(j+1)-m(m\pm 1)]^{1/2}}{4j(j+1)},\\
g_{j,k,m}&:=  \dfrac{km}{j(j+1)}.
\end{align*}
We need to prove that the pairings allowed by the control operators $H_1,H_2$ and $H_3$ conserve the parity of $j+p+k$, when $\delta=(\delta_1,0,0)^T$. To do so, let us compute
\begin{align}
\langle  S_{k,m,p}^j,\ii H_1 S_{k+1,m+1,p}^{j+1}\rangle &=-c_{j,k,m}(\ii\delta_1)+c_{j,k,m}(-\ii\delta_1)\nonumber\\ 
&=-2\ii c_{j,k,m}\delta_1,\label{pairingwang1} \\ 
\langle  S_{k,m,p}^j,\ii H_1 S_{k+1,m+1,p'}^{j+1}\rangle &=-c_{j,k,m}(\ii\delta_1)-c_{j,k,m}(-\ii\delta_1)\nonumber\\
 &=0, \qquad p\neq p',  \nonumber
\end{align}
having used \eqref{eq:wangfunctions},\eqref{kk+1} and the fact that $\delta_2=0$. Then we also have 
\begin{equation}\label{pairingwang2}
\begin{cases}
\langle  S_{k,m,p}^j,\ii H_1 S_{k+1,m+1,p}^{j}\rangle =0, & \\
\langle  S_{k,m,p}^j,\ii H_1 S_{k+1,m+1,p'}^{j}\rangle =-2\ii h_{j,k,m} \delta_1, \quad p\neq p', &
\end{cases}
\end{equation}
having used \eqref{eq:wangfunctions}, \eqref{jj} and the fact that $\delta_2=0$.
The same happens if we replace $m+1$ with $m-1$ and $k+1$ with $k-1$ in \eqref{pairingwang1} and \eqref{pairingwang2}. Because of the selection rules (\ref{rules}), these are the only transitions allowed by the operator $H_1$, and we have thus proved that they conserve the parity of $j+p+k$.
In the same way, one easily checks that every transition induced by $H_2,H_3$ also conserves the parity of $j+p+k$, when $\delta=(\delta_1,0,0)^T$. 

The proof of (iii) is completely analogous to the one of (ii): \eqref{eigenfunctions_bis} implies that $H\mathcal{G}_{e(o)}\subset \mathcal{G}_{e(o)}$ and the claim
 \begin{equation}\label{selection3}
\delta=(0,\delta_2,0)^T\;\Rightarrow\;\langle  S^j_{k,m,p},H_i S^{j'}_{k',m',p'}\rangle =0\quad \text{ if }j+p\not\equiv j'+p'\;\; \mathrm{mod}\;2, \;i=1,2,3,
 \end{equation}
which follows by inspection of the allowed pairings \eqref{kk+1} and \eqref{jj} when $\delta=(0,\delta_2,0)$ as in the proof of claim \eqref{selection2}, plus \eqref{eigenfunctions_bis} imply
$$ \delta=(0,\delta_2,0)^T\;\;\Rightarrow  H_i \mathcal{G}_{e(o)}\subset \mathcal{G}_{e(o)},\;\;i=1,2,3.$$
\end{proof}

\subsection{Controllability of quantum asymmetric tops}
In this section we prove that, except for the dipole configurations (i),(ii), and (iii) of Theorem \ref{thm:symmetries-asymmtop}, system \eqref{schro} is almost always approximately controllable:
\begin{theorem}\label{thm:asymmtopcontrol}
If $\delta \neq (\delta_a,0,0)^T,(0,\delta_b,0)^T,(0,0,\delta_c)^T$, \eqref{schro} is approximately controllable for almost every $A> B> C> 0$.
\end{theorem}
\begin{proof}The proof is based on an application of Theorem \ref{LGTC}. Since the spectrum of $H$ is not explicit if $A> B> C> 0$, we have to check the conditions needed to apply Theorem \ref{LGTC} on system \eqref{schro} in a perturbative way, starting from cases in which the spectrum of $H$ is explicit (that are, $A=B$ or $B=C$).
 We start by considering the case $\delta\neq (\delta_a,\delta_b,0)^T, (0,0,\delta_c)^T$ (that is, the dipole is not orthogonal nor parallel to the axis $c$). The orthogonal case $\delta=(\delta_a,\delta_b,0)^T$ with $\delta_a,\delta_b\neq 0$ will be treated at the end.

We then express the asymmetric Hamiltonian as a perturbation of an oblate symmetric Hamiltonian, using a single asymmetry parameter:
\begin{equation}\label{H(b)}
\begin{split}
H(A,B,C)&=AP_a^2+BP_b^2+CP_c^2\\ &=\frac{1}{2}(A+B)P^2+[C-\frac{1}{2}(A+B)]P_c^2\\ &+\Big(\frac{A-B}{2C-B-A}\Big)[C-\frac{1}{2}(A+B)](P_a^2-P_b^2)\\ &
=:H\Big(\frac{1}{2}(A+B),\frac{1}{2}(A+B),C\Big)+\mu_{\rm o}V_{\rm o}=H(\mu_{\rm o}),
\end{split}
\end{equation}
where $H\Big(\frac{1}{2}(A+B),\frac{1}{2}(A+B),C\Big)$ is the rotational Hamiltonian of an oblate symmetric top with rotational constants $\frac{1}{2}(A+B)$ and $C$, symmetry axis $c$, dipole $\delta=(\delta_a,\delta_b,\delta_c)^T$ with $\delta_a,\delta_c\neq 0$ or $\delta_b,\delta_c\neq 0$ (that is, dipole not parallel nor orthogonal to the symmetry axis $c$) and we have defined the Wang oblate asymmetry parameter \cite[Chapter 7]{gordy}
\begin{equation}\label{eq:wang}
\mu_{\rm o}:= \frac{A-B}{2C-B-A}\in[-1,0],
\end{equation}
and a perturbation operator
$$V_{\rm o}:=[C-\frac{1}{2}(A+B)](P_a^2-P_b^2). $$
We define for every $j \in \mathbb{N}$ the set $I_j:=\{\rho(l,\tau,m)\mid  l=j,j+1, \; \tau,m=-l,\dots,l\}\subset \mathbb{N}$, where $\rho: \{(l,\tau,m) \mid  l\in \mathbb{N},\tau,m=-l,\dots,l\}\rightarrow \mathbb{N}$ is the lexicographic ordering.
The graph $\mathcal{G}$ whose vertices are the sets $I_j$ and whose edges are 
$\{(I_j,I_{j'})\mid  I_j\cap I_{j'} \neq \emptyset\}=\{ (I_j,I_{j+1})\mid  j\in\mathbb{N}\}$ is connected (it is indeed linear). Theorefore, for every $j\in \mathbb{N}$, we consider 
$$\mathcal{M}_j:=\mathcal{H}_j \oplus \mathcal{H}_{j+1},\quad \mathcal{H}_l:= \mathrm{span}\{\Psi_{\tau,m}^l(A,B,C) \mid  \tau,m=-l,\dots,l\},$$
where with a slight abuse of notation we dropped the dependence on $A,B,C$ of the vector spaces $\mathcal{H}_l$ and $\mathcal{M}_j$. The dimension of $\mathcal{M}_j$ is $(2j+1)^2+(2(j+1)+1)^2$, and we identify $\mathfrak{su}(\mathcal{M}_j)$ with $\mathfrak{su}((2j+1)^2+(2(j+1)+1)^2)$. 
We notice that, in addition to the fact that $\mathcal{H}_j$ is by definition invariant under the action of $H(\mu_{\rm o})$, we furthermore have that $\mathcal{H}_j$ is invariant under the actions of $H(0)$ and $V_{\rm o}$, since
\begin{eqnarray}
\langle  D^j_{k,m},H(0)D^{j'}_{k',m'}\rangle &=&0,\quad \text{if } (j,k,m)\neq (j',k',m'), \label{eq:unperturbed_drift}\\
\langle  D^j_{k,m},V_{\rm o}D^{j'}_{k',m'}\rangle &=&0,\quad \text{if } j\neq j',\label{eq:perturbation_rules}
\end{eqnarray}
having used that $\{D^j_{k,m}\}_{j,k,m}$ is a set of orthonormal eigenfunctions for $H(0)$ in \eqref{eq:unperturbed_drift} and \cite[Table 7.2]{gordy} in \eqref{eq:perturbation_rules}.
 We can thus identify $H(\mu_{\rm o})=H(0)+\mu_{\rm o}V_{\rm o}$ with its matrix representation acting on $\mathcal{M}_j=\mathcal{H}_j\oplus\mathcal{H}_{j+1}$: thanks to \eqref{eq:unperturbed_drift} and \eqref{eq:perturbation_rules} the eigenpairs of $H(\mu_{\rm o})$ acting on $L^2({\rm SO}(3))$ and of $H(\mu_{\rm o})$ acting on $\mathcal{M}_j$ are the same, and are thus analytic w.r.t. $\mu_{\rm o}\in[-1,0]$, since $\mathcal{M}_j$ is finite-dimensional. Set 
\[
\begin{split}
&E_{k,0}^j\,:=\,E_{k}^j\;, \quad k=0,\dots,j,\\ &
E_{k,1}^j\,:=\,E_{-k}^j\;, \quad k=1,\dots,j,
\end{split}
\]
being $E^j_k=E^j_k\left(\frac{1}{2}(A+B),C\right)$ defined in \eqref{spectrum}, and let $E^j_{k,0}=E^j_{k,1}$, $k> 0$, be a degenerate eigenvalue of $H(0)$ with two distinct eigenfunctions $S^j_{k,m,0}$ and $S^j_{k,m,1}$, $m=-j,\dots,j$, $k=1,\dots,j$. We can then consider the eigenvalue $E_{k,p}^j(\mu_{\rm o})$ of $H(\mu_{\rm o})$ which converges to $E_{k,p}^j$ as $\mu_{\rm o}$ tends to $0$, for $p=0,1$. It is well known \cite[Chapter 7]{gordy} that $E_{k,p}^j(\mu_{\rm o})$ is still $(2j+1)$-degenerate w.r.t. $m$, but the 2-fold $p$-degeneracy is broken: $E_{k,0}^j(\mu_{\rm o})\neq E_{k,1}^j(\mu_{\rm o})$ if $\mu_{\rm o}\neq 0$.  Moreover, as we have already remarked, the function $[-1,0]\ni \mu_{\rm o}\mapsto E_{k,p}^j(\mu_{\rm o})\in\mathbb{R}$ is analytic.\\

Since for $\mu_{\rm o}=0$ the $p$-degeneracy appears, we need to choose the basis of the eigenspace $\mathcal{E}^j_k:=\mathrm{span}\{D^j_{k,m},D_{-k,m}^j\mid  m=-j,\dots,j\}$ (corresponding to the unperturbed degenerate eigenvalue $E^j_k$)
in which the perturbation $V_{\rm o}$ is diagonal \cite[Chapter 5]{sakurai}. Since \cite[Table 7.2]{gordy}
 \[
  \langle  D_{k,m}^j,V_{\rm o} D^j_{\pm k,m'}\rangle =0, \quad \text{ if } m'\neq m,
 \]
 and
 \begin{equation}\label{degenerateperturbation}
 \langle  D_{k,m}^j,V_{\rm o} D^j_{-k,m}\rangle =\begin{cases}
 0 , & \text{ if } k\neq \pm1,\\
 [C-\frac{1}{2}(A+B)]\frac{j(j+1)}{2}, & \text{ if } k= \pm1,
 \end{cases}
 \end{equation}
we see that the perturbation is not diagonal in the Wigner $D$-basis $\{D^j_{k,m},D_{-k,m}^j\mid  m=-j,\dots,j\}$, but it is diagonal in the Wang basis $\{S^j_{k,m,0},S^j_{k,m,1}\mid  m=-j,\dots j\}$, as a basis of $\mathcal{E}_k^j$, for every $k=0,\dots,j$. In other words, to each asymmetric top eigenfunctions $\Psi^j_{\tau,m}$ is attached one and only one perturbed symmetric top eigenfunction $S^j_{k,m,p}(\mu_{\rm o})$ and we can thus write
 $$\mathcal{H}_l=\mathrm{span}\{S^l_{k,m,p}(\mu_{\rm o})\mid  m=-j,\dots,j,\;k=0,\dots,j,\;p=0,1\}, $$
 for $\mu_{\rm o}\in[-1,0]$. Moreover, as we have already remarked, the function $[-1,0]\ni \mu_{\rm o}\mapsto S^j_{k,m,p}(\mu_{\rm o})\in L^2(SO(3))$ is analytic and $S^j_{k,m,p}(\mu_{\rm o})\rightarrow S^j_{k,m,p}$ as $\mu_{\rm o}\rightarrow 0$. 
 We then express the matrices $H_i^{(j)}:=\Pi_{\mathcal{M}_j}H_i\Pi_{\mathcal{M}_j}$ of the control problem projected on $\mathcal{M}_j$ in this basis that depends on the asymmetry parameter. \\
 
  The spectral gaps in $\mathcal{M}_j$ that we consider are perturbations of symmetric top spectral gaps (cf. figures \ref{lambda(0)} and \ref{transitionsfigure}).
We define:
  \begin{equation}\label{usami-perturbed}
 \begin{split}
&\lambda_{k,0}^j(\mu_{\rm o}):= \mid E_{k+1,0}^{j+1}(\mu_{\rm o})-E_{k,0}^j(\mu_{\rm o})\mid , \;\; k=0,\dots,j, \\ &
\lambda_{k,1}^j(\mu_{\rm o}):= \mid E_{k-1,1}^{j+1}(\mu_{\rm o})-E_{k,1}^j(\mu_{\rm o})\mid , \;\; k=2,\dots,j,\\ & \lambda_{1,1}^j(\mu_{\rm o}):= \mid E_{0,0}^{j+1}(\mu_{\rm o})-E_{1,1}^j(\mu_{\rm o})\mid 
\end{split}
\end{equation}
  \begin{equation}\label{usami-perturbed1}
 \begin{split}
&\rho_{k,0}^j(\mu_{\rm o}):= \mid E_{k-1,0}^{j+1}(\mu_{\rm o})-E_{k,0}^j(\mu_{\rm o})\mid , \;\; k=1,\dots,j, \\ & \rho_{0,0}^j(\mu_{\rm o}):= \mid E_{1,1}^{j+1}(\mu_{\rm o})-E_{0,0}^j(\mu_{\rm o})\mid \\&
\rho_{k,1}^j(\mu_{\rm o}):= \mid E_{k+1,1}^{j+1}(\mu_{\rm o})-E_{k,1}^j(\mu_{\rm o})\mid , \;\; k=1,\dots,j,
\end{split}
\end{equation}
then
 \begin{equation}\label{usami1-perturbed}
 \begin{split}
&\eta_{k,0}^j(\mu_{\rm o}):= \mid E_{k+1,0}^{j}(\mu_{\rm o})-E_{k,0}^j(\mu_{\rm o})\mid , \;\; k=0,\dots,j-1, \\ &
\eta_{k,1}^j(\mu_{\rm o}):= \mid E_{k-1,1}^{j}(\mu_{\rm o})-E_{k,1}^j(\mu_{\rm o})\mid , \;\; k=2,\dots,j,\\ & \eta_{1,1}^j(\mu_{\rm o}):= \mid E_{0,0}^{j}(\mu_{\rm o})-E_{1,1}^j(\mu_{\rm o})\mid 
\end{split}
 \end{equation}
 and finally
 \begin{equation}\label{usami2-perturbed}
 \begin{split}
&\sigma^j_{k,p}(\mu_{\rm o}):= \mid E_{k,p}^{j+1}(\mu_{\rm o})-E_{k,p}^j(\mu_{\rm o})\mid  ,\;\; k=1,\dots,j,\;p=0,1,\\ &\sigma^j_{0,0}(\mu_{\rm o}):= \mid E_{0,0}^{j+1}(\mu_{\rm o})-E_{0,0}^j(\mu_{\rm o})\mid .
\end{split}
\end{equation}
\begin{figure}[ht!]\begin{center}
\includegraphics[width=0.5\linewidth, draft = false]{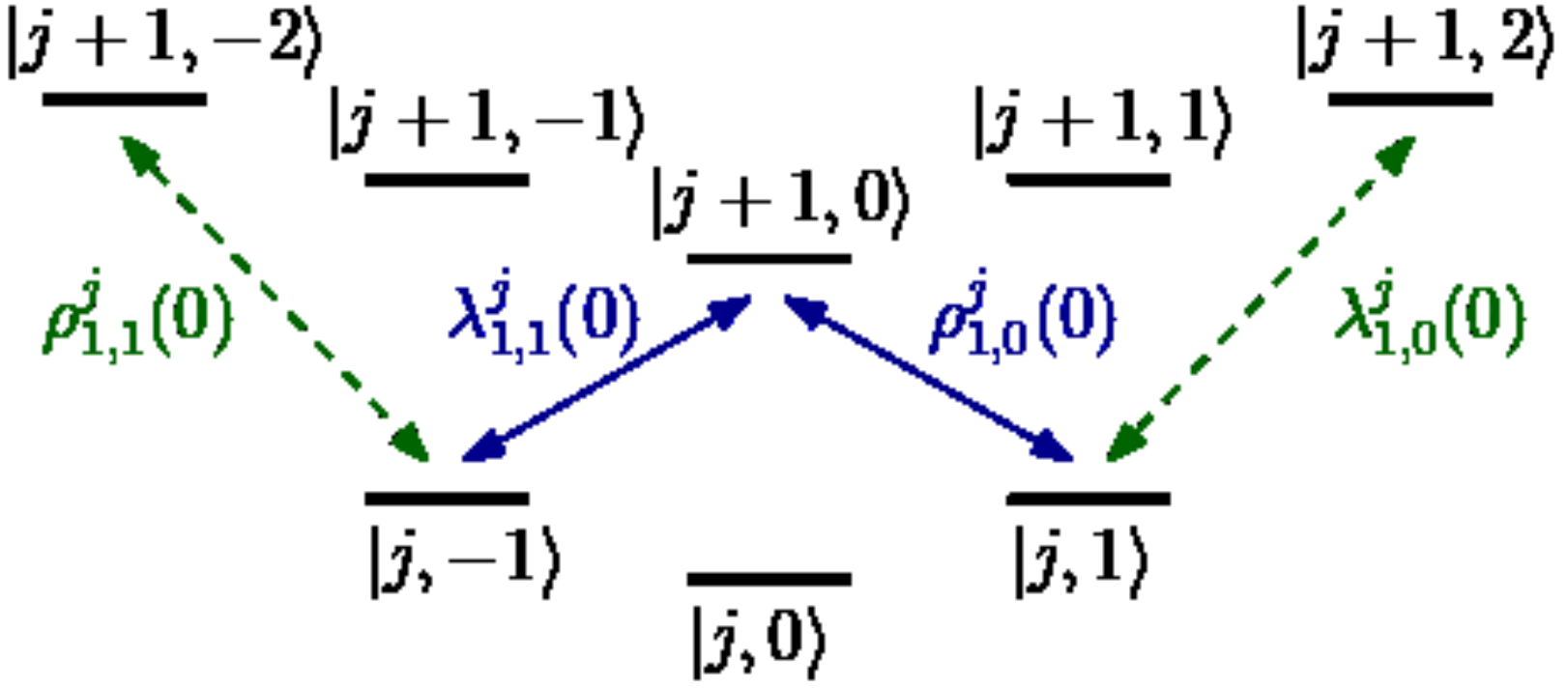}
\caption{Graph of the transitions associated with the unperturbed frequencies $\lambda_{k,p}^j(0)$ and $\rho_{k,p}^j(0)$ between unperturbed eigenstates $\mid  j,k\rangle =\mid  j,k,m\rangle :=D_{k,m}^j$ ($m$ fixed). Same-shaped arrows correspond to equal spectral gaps.} \label{lambda(0)}
\end{center}\end{figure}\;\;\;\;\;\;\;\;
\begin{figure}[ht!]
\subfigure[]{
\includegraphics[width=0.47\linewidth, draft = false]{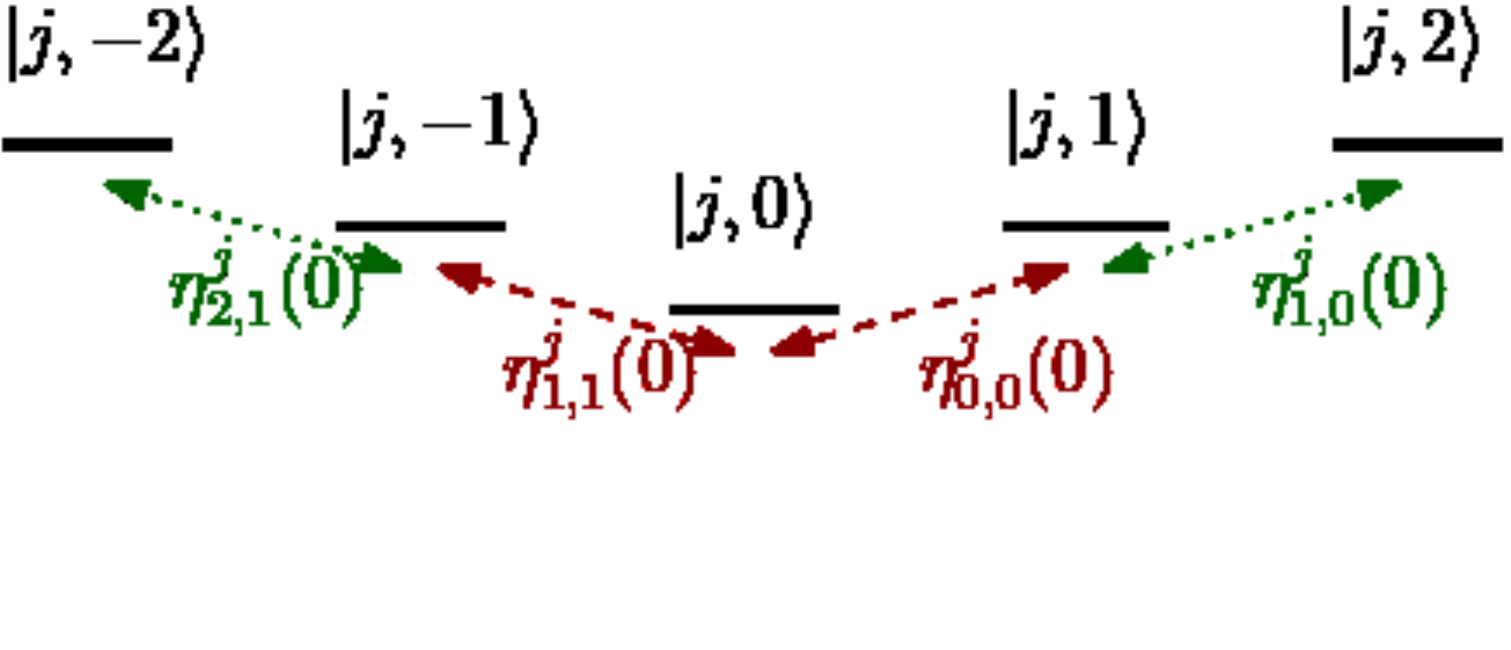} \label{eta(0)} }
\subfigure[]{
\includegraphics[width=0.37\linewidth, draft = false]{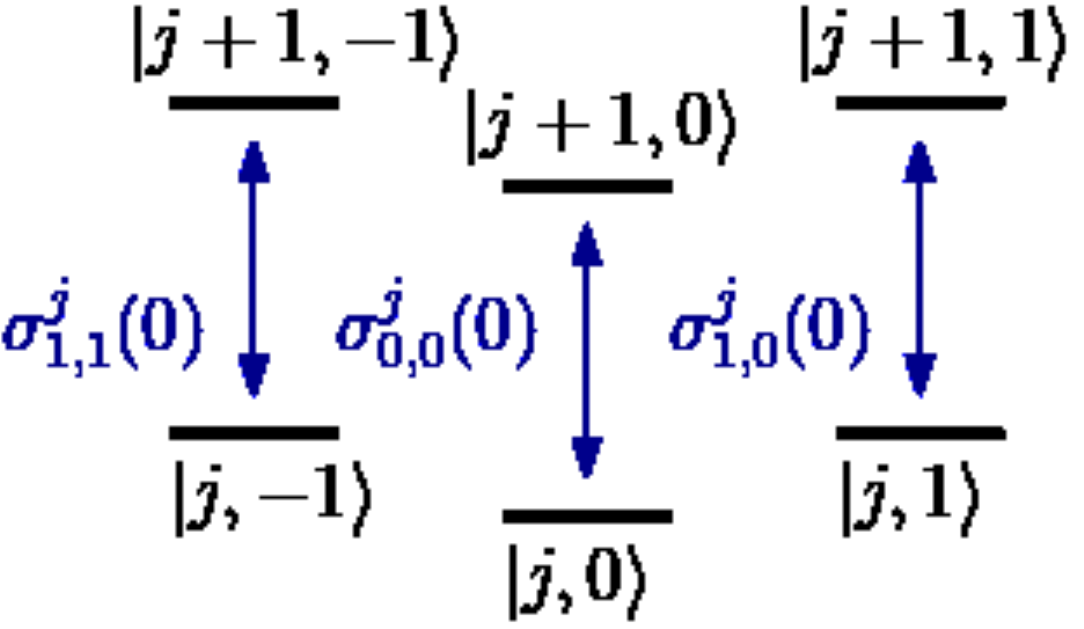} \label{sigma(0)} }
\caption{Transitions between unperturbed eigenstates  $\mid  j,k\rangle =\mid  j,k,m\rangle :=D_{k,m}^j$ ($m$ fixed):
\subref{eta(0)} at unperturbed frequency $\eta^j_{k,p}(0)$; \subref{sigma(0)} at unperturbed frequency $\sigma^j_{k,p}(0)$. Same-shaped arrows correspond to equal spectral gaps.} \label{transitionsfigure}\end{figure}

\begin{remark}
The spectral gaps listed above exhibit several symmetries at $\mu_{\rm o}= 0$, due to the fact that $E^j_{k,0}(0)=E^j_{k,1}(0)$: indeed, we have (compare also with figures \ref{lambda(0)} and \ref{transitionsfigure})
\begin{equation}\label{eq:symmetries}
\begin{split}
&\lambda^j_{0,0}(0)=\rho^j_{0,0}(0)\;,\;\;\lambda^j_{k,0}(0)=\rho^j_{k,1}(0)\;,\;\;\lambda^j_{k,1}(0)=\rho^j_{k,0}(0)\;,\; k=1,\dots,j,\\ &
\eta^j_{k,0}(0)=\eta^j_{k+1,1}(0)\;,k=0,\dots,j-1,\;\;
\sigma^j_{k,0}(0)=\sigma^j_{k,1}(0)\;,\;\;k=1,\dots,j.
\end{split}
\end{equation}
\end{remark}
\begin{lemma}\label{ext-perturbed}
Let $(A+B)/C\notin \mathbb{Q}$. Then, for almost every $\mu_{\rm o}\in[-1,0]$, $(\omega_{ k,p}^j(\mu_{\rm o}),l)\in \Xi_j^0$, for all $\omega\in\{\lambda,\rho,\sigma\}$ for all $k=0,\dots,j$, if $p=0$, and for all $k=1,\dots,j$, if $p=1$, with $l=1,2,3$. Moreover, $(\eta_{k,p}^j(\mu_{\rm o}),l) \in \Xi_j^1$, for all $k=0,\dots,j-1$, if $p=0$ and for all $k=1,\dots,j$, if $p=1$, with $l=1,2,3$.
\end{lemma}
\begin{proof}
We need to prove that, for almost every $\mu_{\rm o}$, $(\omega^j_{k,p}(\mu_{\rm o}),l)$ belongs to $ \Xi_j^1$ for all $ \omega\in\{\lambda,\rho,\eta,\sigma\}$: thanks to the selection rules \eqref{rules}, we have to prove that the spectral gaps $\omega^j_{k,p}(\mu_{\rm o})$ do not couple transitions between two states that are, respectively, in $\mathcal{H}_{j-1}$ and $\mathcal{H}_{j}$ or in $\mathcal{H}_{j+1}$ and $\mathcal{H}_{j+2}$, for almost every $\mu_{\rm o}$. We are thus concerned with the condition
\begin{equation}\label{nonresonance}
\omega^j_{k,p}(\mu_{\rm o}) \neq \omega,\quad \forall \omega \in \Sigma_{j-1,j}\cup\Sigma_{j+1,j+2},
\end{equation}
where $\Sigma_{t,u}:=\{\mid E^t_{k,p}(\mu_{\rm o})-E^u_{k',p'}(\mu_{\rm o})\mid \,\mid  k,k'=0,\dots,j,\;p,p'=0,1\}$ is the set of spectral gaps of $H(\mu_{\rm o})$ between eigenstates in $\mathcal{H}_t$ and $\mathcal{H}_u$.
 Since every spectral gap $\omega$ of $H(\mu_{\rm o})$ can be seen as an analytic function $\omega(\mu_{\rm o})$ where $\omega(0)$ is a spectral gap of $H(0)$, and since the zeros of an analytic function (in this case, the function $\omega^j_{k,p}(\mu_{\rm o})-\omega(\mu_{\rm o})$) on $[-1,0]$ are finite, condition \eqref{nonresonance} is implied for all but a finite number of $\mu_{\rm o}\in[-1,0]$ if it is true at $\mu_{\rm o}=0$. At $\mu_{\rm o}=0$, \eqref{nonresonance} is true if $(A+B)/C \notin \mathbb{Q}$ after \cite[Lemma 3.9]{Ugo-Mario-Io-symmetrictop}. 
 
 The additional requirement  $(\omega^j_{k,p}(\mu_{\rm o}),l)\in \Xi_j^0\; \forall \omega\in\{\lambda,\rho,\sigma\}$, is proved in an analogous way: for each $m,n\in\mathbb{N}\setminus\{j,j+1\}$, we are concerned with the condition
 \begin{equation}\label{nonresonance1}
\omega^j_{k,p}(\mu_{\rm o}) \neq \omega,\quad \forall \omega \in \Sigma_{m,n}
\end{equation}
that holds true at $\mu_{\rm o}=0$, if $(A+B)/C \notin \mathbb{Q}$, after \cite[Lemma 3.9]{Ugo-Mario-Io-symmetrictop}. So, by analyticity, \eqref{nonresonance1} holds true for all but a finite number of $\mu_{\rm o}\in[-1,0]$. Hence, the set $\Lambda$ defined as the set of $\mu_{\rm o}$ such that condition \eqref{nonresonance1} holds true for all $m,n \in\mathbb{N}\setminus\{j,j+1\}$ is given by a countable intersection of sets of full measure $1$\footnote{Indeed, $\Lambda=\bigcap_{m,n\in\mathbb{N}\setminus\{j,j+1\}}\Lambda_{m,n}$ where $\Lambda_{m,n}$ is defined as the set of $\mu_{\rm o}\in[-1,0]$ such that \eqref{nonresonance1} holds true.}, which has measure $1$. This concludes the proof.
 \end{proof}
We first consider the family of decoupled control operators (averaged over the $p$-degeneracies that appear at $\mu_{\rm o}=0$)
\begin{equation}\label{decoupledoperators0}
\begin{split}
\mathcal{F}_j(\mu_{\rm o}):=\Big\{ \ii H^{(j)},&\frac{1}{2}\Big(\mathcal{E}_{\lambda_{0,0}^j(\mu_{\rm o})}(\ii H_l^{(j)})+\mathcal{E}_{\rho_{0,0}^j(\mu_{\rm o})}(\ii H_l^{(j)})\Big),\\ & \frac{1}{2}\Big(\mathcal{E}_{\lambda_{k,0}^j(\mu_{\rm o})}(\ii H_l^{(j)})+\mathcal{E}_{\rho_{k,1}^j(\mu_{\rm o})}(\ii H_l^{(j)})\Big), \\&
 \frac{1}{2}\Big(\mathcal{E}_{\lambda_{k,1}^j(\mu_{\rm o})}(\ii H_l^{(j)})+\mathcal{E}_{\rho_{k,0}^j(\mu_{\rm o})}(\ii H_l^{(j)})\Big)\; \mid \;  l=1,2,3, \; k=1,\dots,j \Big\}.
 \end{split}
\end{equation}
We denote by $\mathrm{L}_j(\mu_{\rm o}):=\mathrm{Lie}(\mathcal{F}_j(\mu_{\rm o}))$ and notice that, thanks to Lemma \ref{ext-perturbed}, $\mathrm{L}_j(\mu_{\rm o})\subset\mathrm{Lie}(\nu_j^0)$ (cf.~\eqref{eq:modes}) for almost every $\mu_{\rm o}$. Then we define the family of matrices
\begin{equation}\label{decoupledoperators1}
\begin{split}
&\mathcal{P}_j(\mu_{\rm o}):=\Big\{\ii H^{(j)}, \frac{1}{2}\Big(\mathcal{E}_{\eta_{k,0}^j(\mu_{\rm o})}(\ii H_l^{(j)})+\mathcal{E}_{\eta_{k+1,1}^j(\mu_{\rm o})}(\ii H_l^{(j)})\Big)\mid \;  l=1,2,3, \; k=0,\dots,j \Big\} \\ &
\cup\Big\{\mathcal{E}_{\sigma_{0,0}^j(\mu_{\rm o})}(\ii H_l^{(j)}), 
 \frac{1}{2}\Big(\mathcal{E}_{\sigma_{k,0}^j(\mu_{\rm o})}(\ii H_l^{(j)})+\mathcal{E}_{\sigma_{k,1}^j(\mu_{\rm o})}(\ii H_l^{(j)})\Big)\;\mid \;  l=1,2,3, \; k=1,\dots,j \Big\},
 \end{split}
\end{equation}
and notice that, by Lemma~\ref{ext-perturbed},  
$\mathcal{P}_j(\mu_{\rm o}) \subset \mathrm{Lie}( \nu_j^1)$ (cf.~\eqref{eq:modes}), for almost every $\mu_{\rm o}$.
Therefore, 
\[
\widetilde{\mathcal{P}}_j(\mu_{\rm o}):=\{A, [B,C] \mid  A,B \in \mathrm{L}_j(\mu_{\rm o}), C \in \mathcal{P}_j(\mu_{\rm o})\} \subset \mathcal{T}_j, \text{ for almost every } \mu_{\rm o}\in[-1,0],
\]
where 
we recall that 
$\mathcal{T}_j$ is the minimal ideal of $ \mathrm{Lie}(\nu_j^1)$ containing $\nu_j^0$.

The next proposition concludes the proof of Theorem \ref{thm:asymmtopcontrol} when $\delta\neq (\delta_a,\delta_b,0)^T$:
\begin{proposition}
For almost every $\mu_{\rm o}\in[-1,0]$, $\mathfrak{su}(\mathcal{M}_j)\subset\mathrm{Lie}(\widetilde{\mathcal{P}}_j(\mu_{\rm o}))$.
\end{proposition}
\begin{proof}
We first claim that, for all but a finite number of $\mu_{\rm o}\in[-1,0]$, one can write
\begin{equation}\label{eq:analyticexpression}
\begin{split}
&\mathcal{E}_{\lambda_{k,0}^j(\mu_{\rm o})}(\ii H_3^{(j)})+\mathcal{E}_{\rho_{k,1}^j(\mu_{\rm o})}(\ii H_3^{(j)})\\ &
=\sum_{\substack{m=-j,\dots,j\\p=0,1}} \langle  S_{k,m,p}^j(\mu_{\rm o}),\ii H_3 S_{k+1,m,p}^{j+1}(\mu_{\rm o})\rangle \;\mid S_{k,m,p}^j(\mu_{\rm o})\rangle \langle  S_{k+1,m,p}^{j+1}(\mu_{\rm o})\mid \\ &
 +\langle  S_{k+1,m,p}^{j+1}(\mu_{\rm o}),\ii H_3 S_{k,m,p}^{j}(\mu_{\rm o})\rangle \;\mid S_{k+1,m,p}^{j+1}(\mu_{\rm o})\rangle \langle  S_{k,m,p}^{j}(\mu_{\rm o})\mid ,
 \end{split}
\end{equation}
where the operator $\mid \psi\rangle \langle \phi\mid $ is the rank-one projector defined by $\mid \psi\rangle \langle \phi\mid \varphi:=\langle \varphi,\psi\rangle \phi$, for all $\psi,\phi,\varphi\in \mathcal{M}_j$. Indeed, it is clear that \eqref{eq:analyticexpression} holds when there are no internal resonances, that is, for all $\mu_{\rm o}\in[-1,0]$ such that
\begin{equation}\label{eq:intperturbed}
 \lambda^j_{k,0}(\mu_{\rm o}), \rho^j_{k,1}(\mu_{\rm o})\neq \omega, \quad \forall \omega\in\Sigma_{j,j+1}\setminus\{\lambda^j_{k,0}(\mu_{\rm o}), \rho^j_{k,1}(\mu_{\rm o})\}.
 \end{equation}
Since the elements of $\Sigma_{j,j+1}\setminus\{\lambda^j_{k,0}(\mu_{\rm o}), \rho^j_{k,1}(\mu_{\rm o})\}$ are finite, and every gap $\omega$ can be written as the analytic perturbation $\omega(\mu_{\rm o})$ of a gap at $\mu_{\rm o}=0$, and since \eqref{eq:intperturbed} holds if $(A+B)/C \notin \mathbb{Q}$ at $\mu_{\rm o}=0$ (cf. \cite[Lemma 3.10]{Ugo-Mario-Io-symmetrictop}), by analyticity \eqref{eq:intperturbed} holds for all but a finite number of $\mu_{\rm o}\in[-1,0]$, and the claim is proved. In particular, when $(A+B)/C \notin \mathbb{Q}$, \eqref{eq:analyticexpression} holds at $\mu_{\rm o}=0$ and gives an analytic expression (that is, the RHS) for $\mathcal{E}_{\lambda_{k,0}^j(\mu_{\rm o})}(\ii H_3^{(j)})+\mathcal{E}_{\rho_{k,1}^j(\mu_{\rm o})}(\ii H_3^{(j)})$, for a.e. $\mu_{\rm o}$. Analogous formulas to \eqref{eq:analyticexpression} hold for the other matrices of \eqref{decoupledoperators0} and \eqref{decoupledoperators1} for a.e. $\mu_{\rm o}$. 
 To conclude, since the statement holds at $\mu_{\rm o}=0$ when $(A+B)/C \notin \mathbb{Q}$ and $\delta\neq (0,0,\delta_c)^T,(\delta_a,\delta_b,0)^T$ (cf. \cite[Proposition 3.12]{Ugo-Mario-Io-symmetrictop}), by analyticity it holds for all but a finite number of $\mu_{\rm o}\in[-1,0]$.
\end{proof}
We now turn to the case $\delta=(\delta_a,\delta_b,0)^T$, with $\delta_a,\delta_b\neq 0$. We then express the asymmetric Hamiltonian as a perturbation of a prolate symmetric Hamiltonian, using a single asymmetry parameter
\begin{equation}\label{H(b)bis}
\begin{split}
H(A,B,C)&=AP_a^2+BP_b^2+CP_c^2\\ &=\frac{1}{2}(C+B)P^2+[A-\frac{1}{2}(C+B)]P_a^2\\ &+\Big(\frac{C-B}{2A-B-C}\Big)[A-\frac{1}{2}(C+B)](P_c^2-P_b^2)\\ &
=H\Big(A,\frac{1}{2}(C+B),\frac{1}{2}(C+B)\Big)+\mu_{\rm p}V_{\rm p}=:H(\mu_{\rm p}),
\end{split}
\end{equation}
where $H\Big(A,\frac{1}{2}(C+B),\frac{1}{2}(C+B)\Big)$ is the rotational Hamiltonian of a prolate symmetric top with rotational constants $A$ and $\frac{1}{2}(C+B)$, symmetry axis $a$, dipole $\delta=(\delta_a,\delta_b,0)$ with $\delta_a,\delta_b\neq 0$ (that is, dipole not parallel nor orthogonal to the symmetry axis $a$) and we have defined the Wang prolate asymmetry parameter \cite[Chapter 7]{gordy}
\begin{equation}\label{eq:wangbis}
\mu_{\rm p}:= \frac{C-B}{2A-B-C}\in[-1,0],
\end{equation}
and a perturbation operator
$$V_{\rm p}:=[A-\frac{1}{2}(C+B)](P_1^2-P_2^2). $$
Since $\delta$ is not parallel nor orthogonal to the axis $a$ (which is the symmetry axis of the associated prolate symmetric top), we can then apply the same proof and conclude that system \eqref{schro} is approximately controllable when $\delta=(\delta_a,\delta_b,0)^T$.
\end{proof}


\textbf{Acknowledgments}\\
The author thanks Ugo Boscain and Mario Sigalotti for having proposed this problem, and enlightening discussions on the subject.

 The project leading to this publication has received funding from the European Union’s Horizon 2020 research and innovation programme under the Marie Sklodowska-Curie grant agreement no. 765267 (QuSCo).

\bibliographystyle{siamplain}
\bibliography{references}

\end{document}